\documentclass[a4paper, 11pt]{amsart}

\usepackage{amsmath}
\usepackage{amsfonts}
\usepackage{amssymb}
\usepackage{mathrsfs}                  
\usepackage{amsthm}
\usepackage{amscd}
\usepackage[hyperfootnotes=false]{hyperref}
\usepackage[margin=4cm]{geometry}
\usepackage{graphicx}
\graphicspath{{figures/}}
\usepackage{epstopdf}
\numberwithin{equation}{section}

\theoremstyle{plain}
\newtheorem{theorem}{Theorem}[section]
\newtheorem{proposition}[theorem]{Proposition}
\newtheorem{lemma}[theorem]{Lemma}
\newtheorem{corollary}[theorem]{Corollary}

\theoremstyle{definition}
\newtheorem{definition}[theorem]{Definition}
\newtheorem{example}[theorem]{Example}

\theoremstyle{remark}
\newtheorem{remark}[theorem]{Remark}

\newtheoremstyle{com}{}{}{\color{blue}}{}{\color{blue}}{}{ }{}
\theoremstyle{com}

\newcommand{\dd}{\,\mathrm{d}}

\newcommand{\E}{\mathbb{E}}

\newcommand{\R}{\mathbb{R}}

\newcommand{\N}{\mathbb{N}}

\renewcommand{\d}{\mathrm{d}}
\renewcommand{\P}{\mathbb{P}}
\newcommand{\Q}{\mathbb{Q}}

\renewcommand{\H}{\mathcal{H}}

\newcommand{\F}{\mathbb{F}}
\newcommand{\Fc}{\mathcal{F}}
\newcommand{\Vc}{\mathcal{V}}
\newcommand{\NN}{\mathcal{NN}}
\newcommand{\Pc}{\mathcal{P}}
\newcommand{\Xc}{\mathcal{X}}
\renewcommand{\Mc}{\mathcal{M}}

\newcommand{\activF}{\sigma}

\title{Deep Hedging}

\begin{document}

\frenchspacing

\author[Buehler]{Hans Buehler}
\address{Hans B\"uhler, J.P. Morgan, London}\footnote{Opinions expressed in this paper are those of the authors, and do not necessarily reflect the view of JP Morgan.}
\email{hans.buehler@jpmorgan.com}

\author[Gonon]{Lukas Gonon}
\address{Lukas Gonon, Eidgen\"ossische Technische Hochschule Z\"urich, Switzerland}
\email{lukas.gonon@math.ethz.ch}

\author[Teichmann]{Josef Teichmann}
\address{Josef Teichmann, Eidgen\"ossische Technische Hochschule Z\"urich, Switzerland}
\email{josef.teichmann@math.ethz.ch}

\author[Wood]{Ben Wood}
\address{Ben Wood, J.P. Morgan, London }
\email{ben.wood@jpmorgan.com}

\date{\today}

\begin{abstract}

  We present a framework for hedging a portfolio of derivatives in the presence of
  market frictions such as transaction costs, market impact, liquidity constraints or risk limits
  using modern deep reinforcement machine learning methods. \\
  
  We discuss how standard reinforcement learning methods can be applied to non-linear reward
  structures, i.e.~in our case convex risk measures.
  As a general contribution to the use of deep learning for stochastic processes,
  we also show in section~\ref{sec:NeuralNets} that the set of constrained trading strategies
  used by our algorithm is large enough to $\epsilon$-approximate any optimal
  solution. \\
    
    Our algorithm can be implemented efficiently even in high-dimensional situations using modern machine learning tools.
  Its structure does not depend on specific market dynamics, and generalizes across hedging instruments including
  the use of liquid derivatives.  Its computational performance is largely invariant in the size of the portfolio
  as it depends mainly on the number of hedging instruments available.\\
    
  We illustrate our approach by showing the effect on hedging
  under transaction costs in a synthetic market driven by the Heston model, where we outperform
  the standard ``complete market" solution.
\end{abstract}

\maketitle
\frenchspacing

\noindent\textbf{Key words and phrases:}  reinforcement learning, approximate dynamic programming, machine learning, market frictions, transaction costs, hedging, risk management, portfolio optimization. \\
\textbf{MSC 2010 Classification: 91G60, 65K99}\\

\noindent

\section{Introduction}

The problem of pricing and hedging portfolios of derivatives is crucial for pricing
risk-management in the financial securities industry. In idealized
frictionless and ``complete market" models, mathematical finance provides,
with risk neutral pricing and hedging, a tractable solution to this
problem. Most commonly, in such models  only the primary asset such as the equity and few additional factors are modeled.
Arguably, the most successful such model for equity models is Dupire's Local Volatility~\cite{Dupire1994}.
For risk management, we will then compute ``greeks" with respect not only to spot,
but also to calibration input parameters such as forward rates and implied volatilities - even if
such quantities are not actually state variables in the underlying model. Essentially, the models are used as a form of 
low dimensional 
interpolation of the hedging instruments. Under complete market assumptions, pricing and risk of a portfolio of derivatives
is linear.

In real markets, though, trading in any instrument is subject to transaction
costs, permanent market impact and liquidity constraints. Furthermore, any trading desk is typically also limited
by its capacity for risk and stress, or more generally capital. This requires traders to overlay the trading
strategy implied by the greeks computed from the complete-market model with
their own adjustments. It also means that pricing and risk are not
linear, but dependent on the overall book: a new trade which reduces the risk in a particular direction can be priced
more favourably. This is called having an ``axe".

The prevalent use of the ``complete market" models  is due to a lack of efficient alternatives; even with the impressive progress
made in the last years for example around super-hedging, there are still few solutions which will scale well over a large
portfolio of instruments, and which do not depend on the underlying market dynamics.\\

Our \emph{deep hedging} approach addresses this deficiency. Essentially, we model the trading decisions in our hedging strategies
as neural networks; their \emph{feature sets} consist not only of prices of our hedging instruments,
but may also contain additional information such as trading signals, news analytics, or past hedging decisions --
quantitative information a human trader might use, in true machine learning fashion.

Such deep hedging
strategies can be described and trained (optimized in classical
language) in a very efficient way, while the respective algorithms
are entirely model-free and do not depend on the
on the chosen market dynamics. That means we can  include market
frictions such as transaction costs, liquidity constraints, bid/ask
spreads, market impact etc, all potentially dependent on the features of the scenario.

The modeling task now amounts to specifying a market scenario generator, a
loss function, market frictions
and trading instruments. This 
approach lends itself well to statistically driven market dynamics. That also means that we do not need to be able to compute greeks of individual derivatives with
a classic derivative pricing model. In fact, we will need no such ``equivalent martingale model". 
\emph{Our approach is greek-free}. Instead,
we can focus our modeling effort on realistic market dynamics and the actual out-of-sample performance of our hedging signal.

High level optimizers then find reasonably
good strategies to achieve good out-of-sample hedging performance under the stated objective. 
In our examples, we are using gradient descent ``Adam" \cite{Kingma2015} mini-batch training for a 
semi-recurrent reinforcement learning problem.
\\

To illustrate our approach, we  will build on ideas from \cite{Ilhan2009} and
\cite{Foellmer2000} and optimize hedging of a portfolio of derivatives under \emph{convex risk measures}. To be able to compare
our results with classic complete market results, we chose in this article to drive the market with a Heston model. We re-iterate that
our algorithm is not dependent on the choice of the model.

To illustrate our algorithm, we investigate the following questions:
\vspace{-1mm}
\begin{itemize}
\item Section~\ref{subsec:benchmark}: How does neural network hedging
  (for different risk-pre\-ferences) compare to the benchmark in a
  Heston model without transaction costs?
\item Section~\ref{subsec:asymptotics}: What is the effect of
  proportional transaction costs on the exponential utility
  indifference price?
\item Section~\ref{subsec:highD}: Is the numerical method scalable to
  higher dimensions?
\end{itemize}
\vspace{-1mm}
Our analysis is based on out-of-sample performance.

To calculate our hedging strategies numerically, we approximate
them by deep neural networks. State-of-the-art
machine learning optimization techniques (see \cite{Goodfellow2016})
are then used to train these networks, yielding a close-to-optimal
\textit{deep hedge}. This is implemented in {\sc Python} using {\sc TensorFlow}.
Under our Heston model,
trading is allowed in both stock and a variance swap. Even experiments
with proportional transaction costs show promising results and the
approach is also feasible  in a high-dimensional setting.

\subsection{Related literature}

There is a vast literature on hedging in market models with
frictions. We only highlight a few to demonstrate the complex
character of the problem.  For example, \cite{Rogers2010} study a
market in which trading a security has a (temporary) impact on its
price. The price process is modelled by a one-dimensional
Black-Scholes model. The optimal trading strategy can be obtained by
solving a system of three coupled (non-linear) PDEs.  In
\cite{Bank2017} a more general tracking problem (covering the
temporary price impact hedging problem) is carried out for a Bachelier
model and a closed form solution (involving conditional expectations
of a time integral over the optimal frictionless hedging strategy) is
obtained for the strategy. \cite{Soner1995} prove that in a
Black-Scholes market with proportional transaction costs, the cheapest
superhedging price for a European call option is the spot price of the
underlying. Thus, the concept of super-replication is of little
interest to practitioners in the one dimensional case. In higher
dimensional cases it suffers from numerical intractability.

It is well known that deep feed forward networks satisfy universal approximation properties, see, e.g.,
\cite{Hornik1991}.
To understand better why they
 are so efficient at approximating hedging strategies, we rely on the very recent and fascinating results of
\cite{Grohs2017}, which can be stated as follows: they quantify the
minimum network connectivity needed to allow approximation of
\emph{all} elements in pre-specified classes of functions to within a
prescribed error, which establishes a universal link between the
connectivity of the approximating network and the complexity of the
function class that is approximated. An abstract framework for
transferring optimal $M$-term approximation results with respect to a
\emph{representation system} to optimal $M$-edge approximation
results for neural networks is established. These transfer results
hold for dictionaries that are \emph{representable by neural
  networks} and it is also shown in \cite{Grohs2017} that a
wide class of representation systems, coined \emph{affine systems},
and including as special cases wavelets, ridgelets, curvelets,
shearlets, $\alpha$-shearlets, and more generally, $\alpha$-molecules,
as well as tensor-products thereof, are re-presentable by neural
networks. These results suggest an explanation for
the ``unreasonable effectiveness'' of neural networks: they
effectively combine the optimal approximation properties of all affine
systems taken together. In our  application of deep hedging
strategies this  means: understanding the relevant input factors
for which the optimal hedging strategy can be written efficiently.

There are several related  applications of reinforcement learning in finance
which have similar challenges, of which we want to highlight two related streams: 
the first is the application to classic portfolio optimization,
i.e.~without options and under the assumption that market prices are available for all 
hedging instruments. As in our setup, this problem requires the use of non-linear objective
functions, c.f.~for example~\cite{cashpf}  or~\cite{rfPf}.
The second promising application of reinforcement learning is in algorithmic trading, where
several authors have shown promising results, e.g.~\cite{AlgoQTrading} and~\cite{AlgoLSTM}
to give but two examples.

The novelty in this article is that we cover derivatives in the first place, and in particular over-the-counter
derivatives which do not have an observable market price.
For example,~\cite{BSQ} covers hedging using Q-learning with only the stock price under Black\&Scholes assumptions
and without transaction cost.

   This puts our article firmly in the realm
of pricing and risk managing a contingent claims in incomplete markets with friction cost.
A general introduction into quantitative finance with a focus on such markets is~\cite{Foellmer2016}.

\subsection{Outline}

The rest of the article is structured as follows. In
Sections~\ref{sec:Setting} and \ref{sec:Hedging} we provide the
theoretical framework for pricing and hedging using convex risk
measures in discrete-time markets with
frictions. Section~\ref{sec:NeuralNets} outlines the parametrization
of appropriate hedging strategies by neural nets and provides
theoretical arguments why it works. In Section~\ref{sec:Heston}
several numerical experiments are performed demonstrating the
surprising feasibility and accuracy of the method.

\section{Setting: Discrete time-market with Frictions}\label{sec:Setting}
Consider a discrete-time financial market with finite time horizon $T$
and trading dates $0=t_0 < t_1 < \ldots < t_n = T$. Fix a  finite\footnote{
The assumption that $\Omega$ is finite is only
  essential for the numerical solution of the optimal hedging problem
  (from Section~\ref{subsec:numericalSol} onwards). Alternatively, we
  could start with arbitrary $\Omega$ and discretize it for the
  numerical solution. If we imposed appropriate integrability
  conditions on all assets and contingent claims, then the results prior to
  section~\ref{subsec:numericalSol} would remain valid for general
  $\Omega$.}
probability space 
$\Omega=\{\omega_1,\ldots,\omega_N\}$ and a probability measure~$\P$ such that
$\P[\{\omega_i\}]>0$ for all $i$.  We define the set of all real-valued random variables
over~$\Omega$ as
$\Xc:=\{X \colon \Omega \to \R \}$.

We denote by~$I_k$ with values in~$\R^r$
any new market information available at time~$t_k$, including market costs and mid-prices of liquid instruments -- typically quoted
in auxiliary terms such as implied volatilities --,
news, balance sheet information, any trading signals, risk limits etc. The process~$I=(I_k)_{k=0,\ldots,n}$ generates the filtration~$\F=(\Fc_k)_{k=0,\ldots,n}$,
i.e.~$\Fc_k$ represents all information available up to~$t_k$. Note that each~$\Fc_k$-measurable random variable
can be written as a function of~$I_0,\ldots,I_k$; this is therefore the richest available feature set for any
decision taken at~$t_k$.

The market contains $d$ hedging instruments with mid-prices given
by an $\R^d$-valued $\F$-adapted stochastic process
$S=(S_k)_{k=0,\ldots,n}$. We do \emph{not} require that there
is an equivalent martingale measure under which $S$ is a martingale.
We stress that our hedging instruments are not simply primary assets such as equities, but
also secondary assets such as liquid options on the former. Some of those hedging instruments are therefore not tradable before
a future point in time (e.g.~an option only listed in 3M with then time-to-maturity of 6M). Such liquidity restrictions
are modeled alongside trading cost below. 

Our portfolio of derivatives which represents our liabilities is an $\Fc_T$ measurable
random variable $Z$. In keeping with the classic literature we may refer to this as
the \emph{contingent claim}, but
we stress that it is meant to represent a portfolio which is a mix of liquid and OTC derivatives.
The maturity~$T$ is  
the maximum maturity of all instruments, at which point
all payments are known.

\noindent
\emph{No classic derivative pricing model will be  needed to valuate~$Z$ or compute Greeks at any point.}

\subsubsection*{Simplifications}

For notational simplicity, we assume
that all intermediate payments are accrued using a (locally) risk-free overnight rate. This essentially means
we may assume that rates are zero and that all payments occur at~$T$. 
We also exclude for the purpose of this article instruments
 with true optionality such as American options. Finally, we also assume that all currency spot exchange happens
 at zero cost, and that we therefore may assume that all instruments settle in our reference currency.\footnote{See \cite{BuRu06}
 for some background on multi-currency risk measures.}

\subsubsection*{Trading Strategies}

In order to hedge a liability~$Z$ at~$T$, we may trade in~$S$ using an $\R^d$-valued $\F$-adapted stochastic
process $\delta=(\delta_{k})_{k=0,\ldots,n-1}$ with~$\delta_k=(\delta^1_k,\ldots,\delta_k^d)$.
Here,~$\delta^i_k$ denotes the agent's holdings of the $i$th asset at time~$t_k$.
We may also define $\delta_{-1}=\delta_n:=0$ for notational convenience. 

We denote by $\H^u$ the unconstrained set of such trading strategies. However, each $\delta_k$ is subject to additional trading constraints. Such restrictions arise due to liquidity, asset availability or trading restrictions.
They are also used to restrict trading in a particular option prior to its availability. In the example above of an option
which is listed in~3M, the respective trading constraints would be~$\{0\}$ until the 3M point. To incorporate these effects, 
we assume that~$\delta_k$ is restricted to a set~$\H_k$ which is given as the image of a continuous, $\Fc_k$-measurable 
map~$H_k:\R^{d(k+1)} \rightarrow \R^d$, i.e.~$\H_k := H_k(\R^{d(k+1)})$. We stipulate that~$H_k(0) = 0$.

Moreover, for an unconstrained strategy~$\delta^u\in\H^u$, we (successively) define with~$(H\circ \delta^u)_k := H_k((H \circ \delta^u)_0,\ldots,(H \circ \delta^u)_{k-1},\delta^u_k)$
its constrained ``projection" into~$\H_k$. 
We denote by~$\H:=(H \circ \H^u) \subset \H^u$ the corresponding non-empty set
of restricted trading strategies.

\begin{example}
	Assume that~$S$ are a range of options and that~$\Vc^{i}_k(S^i_k)$ computes the Black \& Scholes Vega of each option
	using the various market parameters available at time~$t_k$.
	The overall Vega traded with~$\delta_k$ is then~$\Vc_k(\delta_k-\delta_{k-1}) := | \sum_{i=1}^d \Vc^i_k(S^i_k) 
	(\delta^i_k - \delta^i_{k-1}) |$.
	A~liquidity limit of a maximum tradable Vega of~$\Vc_{\mathrm{max}}$ could then be implemented by
	the map:
	\[ 
		H_k(\delta_0,\ldots,\delta_k) := \delta_{k-1} + (\delta_k-\delta_{k-1}) \frac{ \Vc_{\mathrm{max}} }{ \max\{ \Vc_k(\delta_k-\delta_{k-1}), \Vc_{\mathrm{max}} \} } \ .
	\] 
\end{example}

\subsubsection*{Hedging}
All trading is self-financed, so we may also need to inject additional cash~$p_0$ into our portfolio. A negative
cash injection implies we may extract cash.
In a market without transaction costs the agent's wealth at time $T$
is thus given by
$
	-Z + p_0 + (\delta\cdot S)_T 
$,
 where
\[
 (\delta\cdot S)_T := \sum_{k=0}^{n-1} \delta_{k} \cdot
  (S_{k+1}-S_{k}). 
\]

      \noindent
However, we are interested in situations
where trading cost cannot be neglected. We assume that
any trading activity causes costs as follows: if the agent decides to
buy a position $\mathrm{n}\in\R^d$ in $S$ at time $t_k$, then this
will incur cost $c_k(\mathrm{n})$. The total cost of trading a strategy $\delta$
up to maturity is therefore
\[
	C_T(\delta) := \sum_{k=0}^n c_k( \delta_k- \delta_{k-1} ) 
\]
(recall $\delta_{-1}=\delta_n := 0$, the latter of which implies full liquidation in $T$).
\noindent
 The agent's terminal portfolio value at $T$ is therefore
\begin{equation}
  \label{eq:terminalPL} \mathrm{PL}_T(Z,p_0,\delta):=
  -Z + p_0 + (\delta\cdot S)_T-C_T(\delta).\end{equation}

Throughout, we assume that the non-negative adapted cost functions are normalized to
$c_k(0)=0$ and that they are
upper semi-continuous.\footnote{This property is needed in the proof of proposition~\ref{lem:approx}.} In our numerical examples we have
assumed zero transaction costs at maturity.

Our setup includes the following effects:
  \begin{itemize}
  \item Proportional transaction cost: for for $c_k^{i}>0$ define
  		$
  			c_k(\mathrm{n}) := \sum_{i=1}^d c_k^{i} \,S_k^{i} |\mathrm{n}^{i} | 
  		$.
  		
  \item Fixed transaction costs: for $c_k^{i}>0$ and $\varepsilon>0$ set
  		$
  			c_k(\mathrm{n}) := \sum_{i=1}^d c_k^{i} 1_{|\mathrm{n}^{i} |\geq \varepsilon}  			
  		$.

  \item Complex cross-asset cost, such as cost of volatility when trading options across the surface:
  		assume $S^{1}$ is spot and that the rest of the hedging instruments are
  		options on the same asset. Denote by~$\Delta_k^{i}$  Delta
  		and by~$\Vc^{i}_k$  Vega of each instrument, for example under a simple Black \& Scholes model.
  		
  		We may then define a simple cross-surface proportional cost model in Delta and Vega for~$c_k>0$
  		and~$v_k>0$ as
  		\[
  			c_k(\mathrm{n}) :=  c_k^{i} S_k^{1} \left| 1 + \sum_{i=2}^d \Delta^{i}_k   \mathrm{n}^{i} \right| 
  			+
  			v_k^{i} \left| \sum_{i=2}^d  \Vc^{i}_k  \mathrm{n}^{i} \right|
  		\]
 
    \end{itemize} 

\begin{remark}
Our general setup also allows modeling true market impact: in this case, the asset distribution is
affected by our trading decisions. 

As an example for permanent market impact, assume for simplicity that~$I=S$ and
that we have a statistical model of our market in the form of a conditional
distribution $P( S_{k+1} | S_{k} )$. For a proportional impact parameter~$\iota>0$ 
we may now define the dynamics of~$S$ under exponentially
decaying, proportional market impact as $P\left(\ S_{k+1}\ \big|\ S_k \left( 1 + 
	\iota (\delta_k - \delta_{k-1})\right) \ \right)$.
The cost function is accordingly~$c_k(\mathrm{n}) := S_k \iota |\mathrm{n}|$. 

In a similar vein, dynamic  market impact with  decay such as described in~\cite{GaSch2013}
can be implemented.

The real challenge with modeling impact is the effect of trading in one hedging instrument on other hedging instruments,
for example when trading options.
\end{remark}

\section{Pricing and hedging using convex risk measures}\label{sec:Hedging}
In an idealized complete market with continuous-time
trading, no transaction costs, and unconstrained hedging, for any 
liabilities $Z$ there exists a unique replication strategy $\delta$ and a fair price $p_0 \in \R$ such that
$-Z + p_0 + (\delta \cdot S)_T - C_T(\delta) = 0$ holds $\P$-a.s. This is not true in our current setting.

In an incomplete market with frictions, an agent  has to specify an optimality criterion which
defines an acceptable ``minimal price" for any position. 
Such a minimal price is the going to be the minimal amount of cash we need to add to our position in order to implement the optimal 
hedge and such that the overall position becomes acceptable in light of the various costs and constraints.

We focus here on optimality under \emph{convex risk measures} as studied e.g.~in
\cite{Xu2006} and \cite{Ilhan2009}. See also \cite{Kloeppel2007} and
further references therein for a dynamic setting. Convex risk measures
are discussed in great detail in \cite{Foellmer2016}.

\begin{definition}
	Assume that~$X, X_1, X_2 \in \Xc$ represent asset positions (i.e.,~$-X$ is a liability).
	
	\noindent
	We call~$\rho:\Xc\rightarrow\R$ a \emph{convex risk measure} if it is:
	\begin{enumerate}
		\item Monotone decreasing:~if $X_1\geq X_2$ then $\rho(X_1) \leq \rho(X_2)$.
		
			\noindent
			\textit{A more favorable position requires less cash injection}.

		\item Convex: $\rho(\alpha X_1 + (1-\alpha) X_2) \leq \alpha \rho(X_1) + (1-\alpha) \rho(X_2)$ for~$\alpha\in [0,1]$.
		
			\noindent
			\textit{Diversification works}.
			
		\item Cash-Invariant: $\rho(X + c) = \rho(X) - c$ for~$c\in\R$.
		
			\noindent
			\textit{Adding cash to a position reduces the need for more by as much.
			 In particular, this means that~$\rho(X + \rho(X))
			= 0$,
			i.e.~$\rho(X)$ is the least amount~$c$ that needs to be added to 
			the position $X$ in order to make it acceptable in
			the sense that~$\rho(X+c) \leq 0$.}
	\end{enumerate}
	We call~$\rho$ \emph{normalized} if~$\rho(0) = 0$.
\end{definition}

Let $\rho \colon \Xc \to \R$ be such a convex risk measure and for $X \in \Xc$ consider the optimization problem 
\begin{equation}\label{eq:objective1} 
	\pi(X) := \inf_{\delta \in \H}
  \rho\!\left(X + (\delta \cdot S)_T - C_T(\delta)\right) \ .
\end{equation}

\begin{proposition} 
	$\pi$ is monotone decreasing and cash-invariant.
	
	If moreover $C_T(\cdot)$ and $\H$ are convex, then the functional~$\pi$ is a convex risk measure. 
\end{proposition}

\begin{proof}
 For convexity, let~$\alpha\in [0,1]$, set~$\alpha':=1-\alpha$ and assume
	$X_1, X_2 \in \Xc$. Then using the definition of $\pi$ in the first step, convexity of $\H$ in the second step, convexity of $C_T(\cdot)$ combined with monotonicity of $\rho$ in the third step and convexity of $\rho$ in the fourth step, we obtain
	\begin{eqnarray*}
		&& \pi(\alpha X_1 + \alpha' X_2) \\
		&& = \
		 \inf_{\delta \in \H}
  \rho\!\left( \alpha  X_1  + \alpha' X_2 + (\delta \cdot S)_T - C_T(\delta) \right) \\
  && = \
		 \inf_{\delta_1, \delta_2 \in \H}  \rho\!\left( \alpha\left\{  X_1  + (\delta_1 \cdot S)_T \right\} + 
  \alpha' \left\{  X_2  + (\delta_2 \cdot S)_T \right\} -  C_T(\alpha \delta_1 + \alpha' \delta_2 )\right)  \\
  && \leq \
		 \inf_{\delta_1, \delta_2 \in \H}  \rho\!\left( \alpha\left\{  X_1  + (\delta_1 \cdot S)_T - C_T(\delta_1) \right\} + 
  \alpha' \left\{  X_2  + (\delta_2 \cdot S)_T - C_T(\delta_2) \right\}\right) \\
  && \leq \
		 \inf_{\delta_1, \delta_2 \in \H} \left\{ 
  \alpha \rho\!\left( X_1  + (\delta_1 \cdot S)_T - C_T(\delta_1) \right) + 
  \alpha'\rho\!\left( X_2  + (\delta_2 \cdot S)_T - C_T(\delta_2) \right) \right\} \\
  && = \
		 \alpha \pi(X_1) + \alpha' \pi(X_2) \ .
  \end{eqnarray*}	
  Cash-invariance and monotonicity follow directly from the respective properties of~$\rho$. 
\end{proof}

We define an optimal hedging strategy as a minimizer $\delta \in  \H$ of \eqref{eq:objective1}. Recalling the interpretation of 
$\rho(-Z)$ as the minimal amount of capital that has to be added to the risky position $-Z$ to make it acceptable for the risk measure $\rho$, this means that $\pi(-Z)$ is simply the minimal amount that the agent needs to charge in order to make her terminal position acceptable, if she hedges optimally. 

If we defined this as the minimal price, then we would exclude the possibility that having no liabilities may actually have
positive value. This might be the case in the presence of statistically positive expectation of returns under~$\P$ for
some of our hedging instruments. As mentioned before, our framework lends itself to the integration of signals and other
trading information.
We therefore define the \textit{indifference price} $p(Z)$ as the amount of cash that she needs to charge in order to be indifferent between the position $-Z$ and not doing so, i.e. as the solution $p_0$ to $\pi(-Z+p_0)=\pi(0)$.
By cash-invariance this is equivalent to taking $p_0 := p(Z)$, where
\begin{equation}\label{eq:priceDef} p(Z) := \pi(-Z) - \pi(0) \ . 
\end{equation}

It is easily seen that without trading restrictions and transaction costs, this price coincides with the price of a replicating portfolio (if it exists): 
\begin{lemma}  
Suppose $C_T \equiv 0$ and $\H=\H^u$. If $Z$ is
  attainable, i.e. there exists $\delta^* \in \H$ and $p_0 \in \R$ such
  that $Z=p_0+(\delta^* \cdot S)_T$, then $p(Z) = p_0$.
\end{lemma}
\begin{proof} For any $\delta \in \H$, the assumptions and
  cash-invariance of $\rho$ imply
  \[\rho\!\left(-Z + (\delta \cdot S)_T - C_T(\delta)\right) = p_0 + \rho(([\delta
    -\delta^*]\cdot S)_T).\] Taking the infimum over $\delta \in \H$
  on both sides and using $\H - \delta^* = \H$ one obtains
  \[ \pi(-Z)= p_0+ \inf_{\delta \in \H} \rho(([\delta -\delta^*]\cdot    S)_T) = p_0 + \pi(0).\]
\end{proof}

\begin{remark} The methodology developed in this article can also be
  applied to approximate optimal hedging strategies in a setting where
  the price $p_0$ is given exogenously: fix a loss function
  $\ell\colon \R \to [0,\infty)$. Suppose $p_0 > 0$ is given, for example being the result of
  trading derivatives in the market at competitive prices, without taking into
  account risk-management. The agent then wishes to minimize her loss
  at maturity, i.e. she defines an optimal hedging strategy as a
  minimizer to
  \begin{equation}\label{eq:objective2} 
  \inf_{\delta \in \H} \E
    \left[\,\ell(-Z + p_0+(\delta \cdot S)_T - C_T(\delta))\,
    \right].\end{equation}
  This problem, i.e. optimal hedging under a capital constraint, is closely related to taking for $\rho$ a shortfall risk measure, see e.g. \cite{Foellmer2000}. 
\end{remark}

\subsubsection*{Arbitrage}
We mentioned in the introduction that we do not require per se that the market is free of arbitrage.
	To recap, we call~$\delta^{[X]}\in\H$ an arbitrage opportunity given~$X$
	is an opportunity to make money without risk of a loss, i.e.~$0 \leq X + (\delta^{[X]} S)_T - C_T(\delta^{[X]}) =: (*)$
	while $\P[(*) > 0 ] > 0$.

In case such an opportunity exists, we obviously have~$\rho(X) < 0$. 
Depending on the cost function and our constraints~$\H$, we may be able
to invest an unlimited amount into this strategy. In this case, we get~$\pi(X) = -\infty$.
If this applies to~$X=0$, we call such a market \emph{irrelevant}. This is justified by the following observation:

\begin{corollary}
	Assume that~$\pi(0) > -\infty$. Then~$\pi(X) > -\infty$ for all~$X$.
\end{corollary}

\begin{proof}
	Since~$\Omega$ is finite we have~$\sup X<\infty$ and therefore, using monotonicity,~$\pi(X) \geq \pi(\sup X) \geq \pi(0) -\sup X > -\infty$.
\end{proof}

We note, however, that irrelevance is not necessarily a consequence of outright arbitrage; 
such \emph{statistical arbitrage} may also occur in markets without  arbitrage. Consider
to this end the convex risk measure~$\rho(X) := - \E[ X ]$, and assume that the market without interest rates
is driven by a standard Black \& Scholes model with positive drift~$\mu$ between two time points~$t_0$ and~$t_1$, i.e.
\[
	S_0 := 1 \ \ \ \mbox{and} \ \ \ S_1 := \exp\left\{ \mu t_1 + \sigma Z \sqrt{ t_1 } \right\}
\]
for~$Z$ normal and a volatility~$\sigma>0$. Assume the proportional cost of trading~$S$ in~$t_0$ is~$0.5 e^{\mu t_1}$.
In this case~$\rho( \delta_0 S_1 - C_0(\delta ) ) = - 0.5 \delta_0  e^{\mu t_1 }$ for any~$\delta_0\in\R$ which
implies~$\pi(0) = -\infty$. 
Hence, the market is irrelevant, too, 
even if it does not exhibit classic arbitrage. We also note that this is expected in practise: as an example, consider
a strategy which writes options on an underlying. In most market scenarios such a strategy will on average make money,
even if it is subject to potentially drastic short-term losses.

In closing we note that even if the market dynamics exhibit classic arbitrage, and even in the absence of cost
or liquidity constraints, we may not be able to exploit it. Let us assume that for every arbitrage opportunity~$\delta^{[0]}$
there is a
non-zero probability of not making money, i.e.~$\P[(\delta^{[0]} S)_T + C_T(\delta^{[0]}) = 0] > 0$. Under the extreme risk measure~$\rho(X) := -\inf X$
this market remains relevant with~$\pi(0) = 0$.

\subsection{Exponential Utility Indifference Pricing}
The following lemma shows that the present framework includes
exponential utility indifference pricing as studied for example in
\cite{Hodges1989}, \cite{Davis1993},\cite{Whalley1997} and
\cite{Kallsen2015}. Recall that for the exponential utility function
$U(x):=-\exp(-\lambda x), x \in \R$ with risk-aversion parameter
$\lambda >0$ the indifference price $q(Z) \in \R$ of $Z$ is defined
by
\[ \sup_{\delta \in \H} \E \left[U(q(Z)-Z+(\delta \cdot S)_T +
    C_T(\delta)) \right]= \sup_{\delta \in \H} \E \left[U((\delta
    \cdot S)_T + C_T(\delta)) \right].\] In other words, if the seller
charges a cash amount of $q(Z)$, sells $Z$ and trades in the market,
she obtains the same expected utility as by not not selling $Z$ at
all.

\begin{lemma}\label{lem:expIP} Define $q(Z)$ as above. Choose $\rho$
  as the \textit{entropic risk measure}
  \begin{equation}\label{eq:entropic} \rho(X)= \frac{1}{\lambda}\log
    \E[\exp(-\lambda X)], \end{equation}
  and define $p(Z)$ by \eqref{eq:priceDef}. Then $q(Z)=p(Z)$. 
\end{lemma}
\begin{proof}
  Using the special form of $U$, one may write the indifference price
  as
  \[ q(Z) = \frac{1}{\lambda}\log\left(\frac{\sup_{\delta \in \H} \E
        \left[U(-Z+(\delta \cdot S)_T +
          C_T(\delta))\right]}{\sup_{\delta \in \H} \E \left[U((\delta
          \cdot S)_T + C_T(\delta))\right]}\right) \] and so the claim
  follows from \eqref{eq:priceDef} and \eqref{eq:entropic}.
\end{proof}

\subsection{Optimized certainty equivalents}
Assume that $\ell \colon \R \to \R$ is a \emph{loss function}, i.e.~continuous, non-decreasing and convex. 
We may define a convex risk measure
$\rho$ by setting
\begin{equation}\label{eq:OCE}\rho(X):= \inf_{w \in \R} \left\lbrace w
    + \E[\ell(-X-w)] \right\rbrace, \quad X \in
  \Xc. \end{equation}
\begin{lemma}
  \eqref{eq:OCE} defines a convex risk measure.
\end{lemma}
\begin{proof} Let $X, Y \in \Xc$ be assets.
  \begin{itemize}
  \item[(i)]Monotonicity: suppose $X \leq Y$. Since $\ell$ is
    non-decreasing, for any $w \in \R$ one has
    $\E[\ell(-X-w)] \geq \E[\ell(-Y-w)]$ and thus $\rho(X)\geq \rho(Y)$.
  \item[(ii)]Cash invariance: for any $m \in \R$, \eqref{eq:OCE}
    gives
    \[ \rho(X+m)= \inf_{w \in \R} \left\lbrace (w+m)-m +
        \E[\ell(-X-(w+m))] \right\rbrace = -m+\rho(X). \]
  \item[(iii)]Convexity: let $\lambda \in [0,1]$. Then convexity of
    $\ell$ implies
    \[\begin{aligned}\rho(\lambda X + &(1-\lambda)Y) \\ & = \inf_{w
          \in \R} \left\lbrace w + \E[\ell(-\lambda X -(1-\lambda)Y-w)]
        \right\rbrace \\& = \inf_{w_1,w_2 \in \R} \left\lbrace \lambda
          w_1 + (1-\lambda) w_2 + \E[\ell(\lambda
          (-X-w_1)+(1-\lambda)(-Y-w_2))] \right\rbrace \\& \leq
        \inf_{w_1 \in \R}\inf_{w_2 \in \R} \left\lbrace \lambda (w_1+
          \E[\ell(-X-w_1)]) + (1-\lambda)(w_2 + \E[\ell(-Y-w_2)])
        \right\rbrace \\ & =
        \lambda\rho(X)+(1-\lambda)\rho(Y). \end{aligned}\]
  \end{itemize}
  \vspace{-5mm}
\end{proof}

Taking $\ell(x):=-u(-x)$ ($x \in \R$) for a utility function
$u \colon \R \to \R$, \eqref{eq:OCE} coincides with the optimized
certainty equivalent as defined (and studied in a lot more detail than
here) in \cite{BenTal2007}.
\begin{example}\label{ex:entropic} Fix $\lambda > 0$ and set
  $\ell(x):=\exp(\lambda x)-\frac{1+\log(\lambda)}{\lambda}$, $x \in
  \R$. Then the optimization problem in \eqref{eq:OCE} can be solved
  explicitly and the minimizer $w^*$ satisfies
  $e^{\lambda w^*} = \lambda \E[\exp(-\lambda X)]$. Inserting this
  into \eqref{eq:OCE}, one obtains the \textit{entropic risk measure}
  defined in \eqref{eq:entropic} above.

\end{example}

\begin{example}\label{ex:CVAR} Let $\alpha \in (0,1)$ and set
  $\ell(x):=\frac{1}{1-\alpha}\max(x,0)$. The associated risk measure
  \eqref{eq:OCE} is called \textit{average value at risk at level}
  $1-\alpha$ (see \cite[Definition~4.48,
  Proposition~4.51]{Foellmer2016} with $\lambda:=1-\alpha$) or also
  \textit{conditional value at risk} or \textit{expected shortfall}.
\end{example}

\begin{proposition}\label{lem:martingale} Suppose $S$ is a
  $\P$-martingale, $\rho$ is defined as in \eqref{eq:OCE} and $\pi$,
  $p$ as in \eqref{eq:objective1}, \eqref{eq:priceDef}. Then
  \begin{itemize}
  \item[(i)] $\pi(0)=\rho(0)$,
  \item[(ii)] $p(Z) \geq \E[Z]$ for any $Z \in \Xc$.
  \end{itemize}
\end{proposition}
\begin{proof} Since $0 \in \H$ and $C_T(0)=0$, one has
  $\pi(0) \leq \rho(0)$ for any choice of risk measure $\rho$ in
  \eqref{eq:objective1}. Under the present assumptions the converse
  inequality is also true: Since $S$ is a martingale, it holds that
\begin{equation}\label{eq:auxEq7} \E[(\delta \cdot S)_T] = \sum_{j=0}^{n-1} \E\left[\delta_j \E[S_{j+1}-S_j | \Fc_j]\right] = 0 \quad \text{ for any } \delta \in \H. \end{equation} 
  By first applying Jensen's inequality
  (recall that $\ell$ is convex) and then using \eqref{eq:auxEq7},
  that $C_T(\delta) \geq 0$ for any $\delta \in \H$ and that $\ell$ is
  non-decreasing, one obtains
  \begin{equation}\label{eq:OCEPriceEstimate}\begin{aligned} \pi(-Z)& =\inf_{w \in \R} \inf_{\delta \in \H} \left\lbrace w + \E[\ell(Z-(\delta \cdot S)_T +C_T(\delta)-w)] \right\rbrace \\
      & \geq \inf_{w \in \R} \inf_{\delta \in \H} \left\lbrace w + \ell(\E[Z-(\delta \cdot S)_T +C_T(\delta)-w]) \right\rbrace \\
      & \geq \inf_{w \in \R} \left\lbrace w + \ell(\E[Z]-w) \right\rbrace
      = \rho(-\E[Z]) = \E[Z]+\rho(0).
    \end{aligned}
  \end{equation}
  Inserting $Z=0$ yields the converse inequality $\pi(0) \geq \rho(0)$
  and thus (i). Combining (i), \eqref{eq:priceDef} and
  \eqref{eq:OCEPriceEstimate} then directly gives (ii).
\end{proof}

\section{Approximating hedging strategies by deep neural
  networks}\label{sec:NeuralNets}

The key idea that we pursue in this article is to approximate hedging
strategies by neural networks.  Before describing this approach in
more detail we recall the definition and approximation properties of
neural networks and prove some basic results on hedging strategies
built from them. While these results show that the approach is
theoretically well-founded, they are only one reason
 why we have used neural networks (and not some other
parametric family of functions) to approximate hedging strategies. The
other reason is that optimal hedging strategies built from
neural networks can numerically be  calculated very efficiently. This is explained
first for the case of OCE risk measures and for entropic risk. Finally, an extension to
general risk measures is presented.

\subsection{Universal approximation by neural networks}
Let us first recall the definition of a (feed forward) neural network:

\begin{definition}\label{def:nn} Let $L, N_0,N_1,\ldots,N_L \in \N$,
  $\activF \colon \R \to \R$ and for any $\ell=1,\ldots,L$, let
  $W_\ell \colon \R^{N_{\ell-1}} \to \R^{N_\ell}$ an affine function. A
  function $F \colon \R^{N_0} \to \R^{N_L}$ defined as
  \[ F(x)=W_L \circ F_{L-1} \circ \cdots \circ F_1 \text{ with } F_\ell =
    \activF \circ W_\ell \, \text{ for } \ell=1,\ldots,L-1 \] is called a
  (feed forward) neural network. Here the \textit{activation function}
  $\activF$ is applied componentwise. $L$ denotes the number of
  layers, $N_1,\ldots,N_{L-1}$ denote the dimensions of the hidden
  layers and $N_0$, $N_L$ of the input and output layers,
  respectively. For any $\ell=1,\ldots,L$ the affine function $W_\ell$ is
  given as $ W_\ell(x) = A^\ell x + b^\ell$ for some
  $A^\ell \in \R^{N_\ell \times N_{\ell-1}}$ and $b^\ell \in \R^{N_\ell}$. For any
  $i=1,\ldots N_\ell, j=1,\ldots,N_{\ell-1}$ the number $A^\ell_{i j}$ is
  interpreted as the weight of the edge connecting the node $i$ of
  layer $\ell-1$ to node $j$ of layer $\ell$. The number of non-zero weights
  of a network is the sum of the number of non-zero entries of the
  matrices $A^\ell$, $\ell=1,\ldots,L$ and vectors $b^\ell$, $\ell=1,\ldots,L$.
\end{definition}
Denote by $\NN^\activF_{\infty,d_0,d_1}$ the set of neural networks
mapping from $\R^{d_0} \to \R^{d_1}$ and with activation function
$\activF$.  The next result (\cite[Theorems~1 and 2]{Hornik1991})
illustrates that neural networks approximate multivariate functions
arbitrarily well.
\begin{theorem}[Universal approximation, \cite{Hornik1991}]
  \label{thm:universalApprox} Suppose $\activF$ is bounded and
  non-constant. The following statements hold:
  \begin{itemize}
  \item For any finite measure $\mu$ on
    $(\R^{d_0},\mathcal{B}(\R^{d_0}))$ and $1 \leq p < \infty$, the
    set $\NN^\activF_{\infty,d_0,1}$ is dense in $L^p(\R^{d_0},\mu)$.
  \item If in addition $\activF \in C(\R)$, then
    $\NN^\activF_{\infty,d_0,1}$ is dense in $C(\R^{d_0})$ for the
    topology of uniform convergence on compact sets.
  \end{itemize}
\end{theorem}
Since each component of an $\R^{d_1}$-valued neural network is an
$\R$-valued neural network, this result easily generalizes to
$\NN^\activF_{\infty,d_0,d_1}$ with $d_1>1$, see also \cite{Hornik1991}. A
variety of other results with different assumptions on $\activF$ or
emphasis on approximation rates are available, see
e.g. \cite{Grohs2017} for further references.

In what follows, we fix an activation function $\activF$ and omit it in the notation, i.e. we write $\NN_{\infty,d_0,d_1}:=\NN^\activF_{\infty,d_0,d_1}$. Furthermore, we denote
by $\{\NN_{M,d_0,d_1}\}_{M \in \N}$ a sequence of subsets of
$\NN_{\infty,d_0,d_1}$ with the following properties:
\begin{itemize}
\item $\NN_{M,d_0,d_1} \subset \NN_{M+1,d_0,d_1}$ for all $M \in \N$,
\item $\bigcup_{M \in \N} \NN_{M,d_0,d_1} = \NN_{\infty,d_0,d_1}$,
\item for any $M \in \N$, one has
  $\NN_{M,d_0,d_1} = \{F^\theta \colon \theta \in \Theta_{M,d_0,d_1}
  \}$ with $\Theta_{M,d_0,d_1} \subset \R^{q}$ for some $q \in \N$
  (depending on $M$).
\end{itemize}

\begin{remark}\label{rmk:NNex}
  We have two classes of examples in mind: the first one is to take
  for $\NN_{M,d_0,d_1}$ the set of all neural networks in
  $\NN_{\infty,d_0,d_1}$ with an arbitrary number of layers and
  nodes, but at most $M$ non-zero weights. The second one is to take
  for $\NN_{M,d_0,d_1}$ the set of all neural networks in
  $\NN_{\infty,d_0,d_1}$ with a \textit{fixed architecture},
  i.e. a fixed number of layers $L^{(M)}$ and fixed input and output
  dimensions for each layer. These are specified by $d_0$, $d_1$ and
  some non-decreasing sequences $\{L^{(M)}\}_{M \in \N}$ and
  $\{N_1^{(M)}\}_{M \in \N}$, $\ldots$,
  $\{N_{L^{(M)}-1}^{(M)}\}_{M \in \N}$.  In both cases the set
  $\NN_{M,d_0,d_1}$ is parametrized by matrices $A^\ell$ and vectors~$b^\ell$.
\end{remark}

\subsection{Optimal hedging using deep neural networks}\label{subsec:nnHedgingStrateiges}

Motivated by the universal approximation results stated above, we now
consider neural network hedging strategies. Let our activation function therefore be bounded and non-constant.

 In order to apply our theorem \ref{thm:universalApprox}, we represent the optimization over
constrained trading strategies~$\delta\in\H$ as an optimization over~$\delta \in \H^u$ with a
following modified objective.
\begin{lemma}
We may write the constrained problem~\ref{eq:objective1} as the modified unconstrained problem as
\begin{equation}\label{eq:objective1_1} \tag{3.1'}
	\pi(X) = \inf_{\delta \in \H^u}
  \rho\!\left(X + ( H\circ \delta \cdot S)_T - C_T(H \circ \delta)\right)  \ .
\end{equation}
\end{lemma}
\begin{proof}
	Note that $H\!\circ\!\delta = \delta$ for all~$\delta \in \H$, and
	$H\!\circ\!\delta^u\in\H$ for all~$\delta^u\in\H^u$.
\end{proof}
   
     Recall that the information available in our market at $t_k$ is described 
by the observed maximal feature set~$I_0,\ldots,I_k$.
Our trading strategies should therefore depend on this information and on our previous
position in our tradable assets. This gives rise to the following semi-recurrent deep neural network structure
for our unconstrained trading strategies:
\begin{align} \label{eq:hedgingNNFull} 
\H_{M} & = \lbrace
  ( \delta_{k})_{k=0,\ldots,n-1} \in \H^u \, : \, \delta_{k} =
                                                F_k(I_{0},\ldots,I_{k},\delta_{k-1}) \, , F_k \in \NN_{M,r(k+1)+d,d} 
                                                \rbrace \\
  \nonumber & = \lbrace (\delta_{k}^\theta)_{k=0,\ldots,n-1} \in \H^u
              \, : \, \delta_{k}^\theta =
              F^{\theta_k}(I_{0},\ldots,I_{k}, \delta_{k-1}) \, , \theta_k \in
              \Theta_{M,r(k+1)+d,d} \rbrace
\end{align} 
We now replace the set $\H^u$ in \eqref{eq:objective1_1} by
$\H_{M} \subset \H^u$. 
We aim at calculating
\begin{eqnarray}\label{eq:objective1NN}
\pi^M(X) 
  & := & 
  	\inf_{\delta \in \H_M } \rho(X + (H\circ\delta \cdot S)_T - C_T(H\circ\delta)) 
  	\\
  & \nonumber = & 
  \inf_{\theta \in \Theta_M } \rho(X+ ( H\circ \delta^\theta \cdot S)_T - C_T(H\circ \delta^\theta)),
\end{eqnarray}
where $\Theta_M = \prod_{k=0}^{n-1} \Theta_{M,r(k+1)+d,d}$.  Thus, the
infinite-dimensional problem of finding an optimal hedging strategy is
reduced to the finite-dimensional constraint problem of finding optimal
parameters for our neural network.

\begin{remark}
Our setup becomes truly ``recurrent" if we enforce~$\theta^k = \theta^0$ for all~$k$ and add~``$k$"
as a parameter into the network. Below proof applies with few modifications.
\end{remark}

\begin{remark}\label{rmk:recurrent}
   If $S$ is an $(\F,\P)$-Markov process
  and $Z=g(S_T)$ for $g\colon \R^d \to \R$ and with simplistic market frictions we may know
  that the optimal strategy in \eqref{eq:objective1} is of the simpler form
  $\delta_{k} = f_k(I_{k},\delta_{k-1})$ for some
  $f_k \colon \R^{r+d} \to \R^d$.
\end{remark}

\begin{remark}
 We would similarly transform
  \eqref{eq:objective2} into a modified unconstrained problem,
  optimized over~$\H_M$.
\end{remark}

\begin{remark}
For practical implementations, handling trading constraints with~\ref{eq:objective1NN} is not particularly
efficient since the gradient of~$\Theta_M$ of our objective outside~$\H$ vanishes.
In the case where~$H \circ \delta = \delta$ for~$\delta \in \H$, this can be addressed by variants of
\[
	\pi(X) \equiv \inf_{\delta \in \H^u}
  \rho\!\left(X + ( H\circ \delta \cdot S)_T - C_T(\delta) - \gamma \| \delta - H \circ \delta \|_1 \right)   \ .
\]
for Lagrange multipliers~$\gamma \gg 0$.
\end{remark}

The next proposition shows that thanks to
the universal approximation theorem, strategies in $\H$ are
approximated arbitrarily well by strategies in $\H_M$. Consequently,
the neural network price $\pi^M(-Z)-\pi^M(0)$ converges to the exact
price $p(Z)$.

\begin{proposition}\label{lem:approx} Define $\H_M$ as in
  \eqref{eq:hedgingNNFull} and $\pi^M$ as in
  \eqref{eq:objective1NN}. Then for any $X \in \Xc$,
  \[\lim_{M \to \infty} \pi^M(X) = \pi(X) \ .\]
\end{proposition}
\begin{proof}
	We first note that the argument~$\delta_{k-1}$ in~\ref{eq:objective1NN} 
	is redundant, since iteratively $\delta_{k-1}$ is itself a function of~$I_0,\ldots,I_{k-1}$.
	We may therefore write for the purpose of this proof
\begin{equation} \label{eq:hedgingNNFull2}   \tag{4.1'}
\H_{M} = \lbrace (\delta_{k}^\theta)_{k=0,\ldots,n-1} \in \H^u
              \, : \, \delta_{k}^\theta =
              F_k(I_{0},\ldots,I_{k}) \, , F_k \in \NN_{M,r(k+1),d} 
              \rbrace \ .
\end{equation}

    Since $\H_M \subset \H_{M+1} \subset \H^u$ for all
  $M \in \N$ it follows that $\pi^{M}(X) \geq \pi^{M+1}(X) \geq
  \pi(X)$. Thus it suffices to show that for any $\varepsilon > 0$
  there exists $M \in \N$ such that
  $\pi^M(X) \leq \pi(X)+\varepsilon$. 
  
  By definition, there exists
  $\delta  \in \H^u$ such that
  \begin{equation}\label{eq:auxEq1} \rho(X + (H\circ \delta \cdot S)_T -
    C_T(H\circ\delta)) \leq \pi(X)+\frac{\varepsilon}{2}. 
    \end{equation}

   Since~$\delta_k$ is~$\Fc_k$-measurable, there exists $f_k\colon \R^{r(k+1)} \to \R^d$ measurable such that
   $\delta_{k} = f_k(I_{0},\ldots,I_{k})$ for each $k=0,\ldots,n-1$. 
   Since $\Omega$ is finite, $\delta_{k}$ is bounded and so $f^i_k \in L^1(\R^{r(k+1)},\mu)$ for any $i=1,\ldots d$, where $\mu$ is the law of $(I_{0},\ldots,I_{k})$ under $\P$. Thus one may use 
  theorem~\ref{thm:universalApprox} to find ${F}^i_{k,n} \in \NN_{\infty,r(k+1),1}$ such that ${F}^i_{k,n}(I_{0},\ldots,I_{k})$ converges to $f^i_k(I_{0},\ldots,I_{k})$ in $L^1(\P)$ as $n \to \infty$.

      By passing now to a suitable subsequence, convergence holds $\P$-a.s.~simultaneously 
  for all $i,k$. Writing $\delta_k^n:=F_{k,n}(I_0,\ldots,I_k)$ and using $\P[\{\omega\}]>0$ for all $\omega \in \Omega$, this implies
  \begin{equation}\label{eq:auxEq10} 
  \lim_{n \to \infty} \delta_k^n(\omega)  =
  \delta_k(\omega) \quad \text{ for all } \omega \in \Omega. 
  \end{equation}
  Continuity of~$H_k(\cdot)(\omega)$ 
  for a fixed~$\omega$ implies moreover that also~$\lim_{n \to \infty} H_k(\omega) \circ \delta_k^n(\omega)  =
  H_k(\omega) \circ  \delta_k(\omega)$.
  
  Since $\Omega$ is finite, $\rho$ can be viewed as a convex function $\rho \colon \R^N \to \R$. In particular, $\rho$ is continuous. Using
  continuity of $\rho$ in the first step and upper semi-continuity of $c_k(\cdot)(\omega)$ for each $\omega \in \Omega$ combined with monotonicity of $\rho$ in the second step, one obtains
  \begin{eqnarray*}
  && \liminf_{n \to \infty} \rho(X + (H\circ \delta^n \cdot S)_T - C_T(H\circ \delta^n))
  \\
  && \ \ \leq \
  	 \rho(X + (H\circ\delta \cdot S)_T - \limsup_{n \to \infty} C_T(H\circ\delta^n)) 
  \\
  && \ \ \leq \
   \rho(X + (H\circ\delta \cdot S)_T - C_T(H\circ\delta)).
    \end{eqnarray*} 
    Combining this with \eqref{eq:auxEq1}, there
  exists $n \in \N$ (large enough) such that
  \begin{equation}\label{eq:auxEq2} \rho(X + (H\circ\delta^n \cdot S)_T -
    C_T(H\circ\delta^n)) \leq \pi(X)+\varepsilon. 
\end{equation}
  Since $\delta^n \in \H_M$ for all $M$ large enough, one obtains 
  $\pi^M(X)  \leq \pi(X)+\varepsilon$ by \eqref{eq:objective1NN} and \eqref{eq:auxEq2}, as desired. 
\end{proof}

\subsection{Numerical solution for OCE-risk  measures}\label{subsec:numericalSol}

While Theorem~\ref{thm:universalApprox} and
Proposition~\ref{lem:approx} give a theoretical justification for
using hedging strategies built from neural networks, we now turn to
computational considerations: how can we calculate a (close-to)
optimal parameter $\theta \in \Theta_M$ for \eqref{eq:objective1NN}?

To explain the key ideas we focus on the case when $\rho$ is an OCE
risk measure (see \eqref{eq:OCE}) and no trading constraints are present, the case of general risk measures is treated below.

Inserting the definition of $\rho$, see \eqref{eq:OCE}, into
\eqref{eq:objective1NN}, the optimization problem can be rewritten as
\[ 
\pi^M(-Z) = \inf_{\bar{\theta} \in \Theta_M}\inf_{w \in \R}
  \left\lbrace w + \E[\ell(Z - (\delta^{\bar{\theta}} \cdot S)_T +
    C_T(\delta^{\bar{\theta}})-w)] \right\rbrace = \inf_{\theta \in
    \Theta} J(\theta), \] where $\Theta = \R \times \Theta_M$ and for
$\theta = (w,\bar{\theta}) \in \Theta$,
\begin{equation}\label{eq:Jdef} J(\theta):= w+\E[\ell(Z -
  (\delta^{\bar{\theta}} \cdot S)_T +
  C_T(\delta^{\bar{\theta}})-w)]. \end{equation}
Generally, to find a local minimum of a differentiable function $J$, one may use a \textit{gradient descent} algorithm: Starting with an initial guess $\theta^{(0)}$, one iteratively defines 
\begin{equation}\label{eq:gradientDescent} \theta^{(j+1)} =
  \theta^{(j)} - \eta_j \nabla J_j(\theta^{(j)}), \end{equation}
for some (small) $\eta_j > 0$, $j \in \N$ and with $J_j = J$.
Under suitable assumptions on $J$ and the sequence $\{\eta_j\}_{j \in \N}$, $\theta^{(j)}$ converges to a local minimum of $J$ as $j \to \infty$. Of course, the success and feasibility of this algorithm crucially depends on two points: Firstly, can one avoid finding a local minimum instead of a global one? Secondly, can $\nabla J$ be calculated efficiently?

One of the key insights of deep learning is that for cost functions
$J$ built based on neural networks both of these problems can be dealt
with simultaneously by using a variant of \textit{stochastic gradient
  descent} and the \textit{(error) backpropagation} algorithm. What
this means in our context is that in each step $j$ the expectation in
\eqref{eq:Jdef} (which is in fact a weighted sum over all elements of
the finite, but potentially very large sample space $\Omega$) is
replaced by an expectation over a randomly (uniformly) chosen subset
of $\Omega$ of size $N_{\text{batch}} \ll N$, so that $J_j$ used in
the update \eqref{eq:gradientDescent} is now given as
\[ J_j(\theta) = w+ \sum_{m=1}^{N_{\text{batch}}} \ell(Z(\omega_m^{(j)})
  - (\delta^{\bar{\theta}} \cdot S)_T(\omega_m^{(j)}) +
  C_T(\delta^{\bar{\theta}})(\omega_m^{(j)})-w)
  \frac{N}{N_{\text{batch}}} \P[\{\omega_m^{(j)}\}] \] for some
$\omega_1^{(j)},\ldots, \omega_{N_{\text{batch}}}^{(j)} \in
\Omega$. This is the simplest form of the (minibatch) stochastic
gradient algorithm. Not only does it make the gradient computation a
lot more efficient (or possible at all, if $N$ is large), but it also
avoids getting stuck in local minima: even if $\theta^{(j)}$ arrives
at a local minimum at some $j$, it moves on afterwards (due to the
randomness in the gradient). In order to calculate the gradient of
$J_j$ for each of the terms in the sum, one may now rely on the
compositional structure of neural networks. If $\ell$,
$c$ and $\activF$ are sufficiently differentiable and the
derivatives are available in closed form, then one may use the chain
rule to calculate the gradient of $F^{\bar{\theta}_k}$ with respect to
$\theta$ analytically and the same holds for the gradient of $J_j$.
Furthermore, these analytical expressions can be evaluated very
efficiently using the so called backpropagation algorithm (see
subsequent section).

While this certainly answers the second question posed above
(efficiency), the first one (local minima) is only partially resolved,
as there is no general result guaranteeing convergence to the global
minimum in a reasonable amount of time. However, it is common belief
that for sufficiently large neural networks, it is possible to arrive
at a sufficiently low value of the cost function in a reasonable
amount of time, see \cite[Chapter~8]{Goodfellow2016}.

Finally, note that for the experiments in Section~\ref{sec:Heston}
below we have used Adam, a more refined version of the stochastic
gradient algorithm, as introduced in \cite{Kingma2015} and also
discussed in \cite[Chapter~8.5.3]{Goodfellow2016}.

\begin{remark} In the experiments in Section~\ref{sec:Heston} below,
  the functions $\ell$, $c$ and $\activF$ are
  continuous, but have only piecewise continuous
  derivatives. Nevertheless, similar techniques can be applied.
\end{remark}

\begin{remark} Numerically, trading constraints can be handled by introducing Lagrange-multipliers or by imposing infinite trading cost outside the allowed trading range. Certain types of constraints can also be dealt with by the choice of activation function: for example, no short-selling constraints can be enforced by choosing a non-negative activation function $\activF$. A systematic numerical treatment will be left for future research. 
\end{remark}

\subsection{Certainty Equivalent of Exponential Utility}
The entropic risk measure \eqref{eq:entropic} is a special case of an OCE risk measure, as explained in example~\ref{ex:entropic}. However, when applying the methodology explained in Section~\ref{subsec:numericalSol}, there is no
need to minimize over $w$: we may directly insert \eqref{eq:entropic} into \eqref{eq:objective1NN} to write
\[
\pi^M(-Z) =  \frac1\lambda \log  \inf_{\theta \in
    \Theta_M} J(\theta), 
    \]
        where 
\begin{equation}\label{eq:JEntropic} J(\theta):= \E[\ \exp(-\lambda [-Z +
  (\delta^{\theta} \cdot S)_T -
  C_T(\delta^{\theta})])\ ]. \end{equation}
A close-to-optimal $\theta \in \Theta_M$ can then be found numerically as above.

\subsection{Extension to general risk measures}
As explained in Section~\ref{subsec:numericalSol}, for OCE risk
measures the optimal hedging problem \eqref{eq:objective1NN} is
amenable to deep learning optimization techniques (i.e. variants of
stochastic gradient descent) via \eqref{eq:Jdef}. The key ingredient
for this is that the objective $J$ satisfies
\begin{itemize}
\item [(ML1)] the gradient of $J$ decomposes into a sum over the
  samples, i.e.
  $\nabla_\theta J(\theta) = \sum_{m=1}^N \nabla_\theta
  J(\theta,\omega_m) $ and
\item [(ML2)] $\nabla_\theta J(\theta,\omega_m)$ can be calculated
  efficiently for each $m$, i.e. using backpropagation.
\end{itemize}

The goal of the present section is to show that for a general class of
convex risk measures (including all coherent ones) one can approximate
\eqref{eq:objective1} by a minimax problem over neural networks and
that the objective functional of this approximate problem also has
these two key properties, making it amenable to deep learning
optimization techniques.

Denote by $\Pc$ the set of probability measures on $(\Omega,\Fc)$. The
following result serves as a starting point:
\begin{theorem}[Robust representation of convex risk measures] Suppose
  $\rho \colon \Xc \to \R$ is a convex risk measure. Then
  $\rho$ can be written as
  \begin{equation}\label{eq:rhoRobustFS} \rho(X)= \max_{\Q \in
      \Pc}\left(\E_\Q[-X]-\alpha(\Q)\right), \quad X \in
    \Xc,\end{equation}
  where $\alpha(\Q):= \sup_{X \in \Xc}\left( \E_\Q[-X] - \rho(X) \right).$
\end{theorem}
\begin{proof} Since for $\Omega$ finite the set of probability
  measures $\Pc$ coincides with the set of finitely additive,
  normalized set functions (appearing in
  \cite[Theorem~4.16]{Foellmer2016}), the present statement follows
  directly from the cited theorem and
  \cite[Remark~4.17]{Foellmer2016}.
\end{proof}
The function $\alpha \colon \Pc \to \R$ is called the (minimal)
penalty function of the risk measure $\rho$.

Since $\Omega$ is finite, $\Pc$ can be identified with the standard
$N-1$ simplex in $\R^N$ and so \eqref{eq:rhoRobustFS} is an
optimization over $\R^N$. However, $N$ is very large in our context
and so the representation \eqref{eq:rhoRobustFS} is of little use for
numerical calculations. The next result shows that $\rho(X)$ can be
approximated by an optimization problem over a lower-dimensional
space. To state it, let us define the set
$\mathcal{L} \subset \Xc$ of log-likelihoods by
\[ \mathcal{L}:= \{ f \in \Xc \; : \; \E[\exp(f)]=1\}, \]
define $\bar{\alpha} \colon \mathcal{L} \to \R$ by
$\bar{\alpha}(f)=\alpha(\exp(f) \d \P)$ for any $f \in \mathcal{L}$
and write $\Pc_{eq}$ for the set of probability measures on
$(\Omega,\Fc)$, which are equivalent to $\P$. Furthermore, one may
view $\bar{I}=(I_{0},\ldots,I_{n})$ as a map
$\Omega \to \R^{r(n+1)}$.

\begin{theorem}\label{thm:NNriskmeas}
  Suppose
  \begin{itemize} \item[(i)] $\alpha(\Q) < \infty$ for some
    $\Q \in \Pc_{eq}$,
  \item[(ii)] $\bar{\alpha}$ is continuous,
  \item[(iii)] $\Fc = \Fc_T$.
  \end{itemize}
  Then for any $X \in \Xc$,
  $\rho(X)=\lim_{M \to \infty} \rho^M(X)$, where
  \begin{equation}\label{eq:rhoNNDef}
    \rho^M(X):= \sup_{\substack{\theta \in \Theta_{M,r(n+1),1} \\ \E[\exp(F^\theta\circ\bar{I})] = 1}}\left(\E[-X\exp(F^\theta\circ \bar{I})]-\bar{\alpha}(F^\theta \circ \bar{I})\right).
  \end{equation}
\end{theorem}

\begin{proof}
  We proceed in two steps. In a first step we show that for any
  $X \in \Xc$ one may write
  \begin{equation} \label{eq:rhoRepresentation} \rho(X) =
    \sup_{\substack{\bar{f} \in \mathcal{M} \\
        \E[\exp(\bar{f}\circ\bar{I})] = 1}}
    \left(\E[-X\exp(\bar{f}\circ \bar{I})]-\bar{\alpha}(\bar{f}\circ
      \bar{I})\right),
  \end{equation}
  where $\Mc$ denotes the set of measurable functions mapping from
  $\R^{r(n+1)} \to \R$.  In the second step we rely on
  \eqref{eq:rhoRepresentation} to prove the statement.

  \textit{Step 1:} Since $\P[\{\omega_i\}]>0$ for all $i$,
  $\Xc$ coincides with $L^\infty(\Omega,\mathcal{F},\P)$ and
  $\rho$ is law-invariant. Thus by (i) and
  \cite[Theorem~4.43]{Foellmer2016} one may write
  \begin{equation}\label{eq:auxEq3} \rho(X)= \sup_{\Q \in
      \Pc_{eq}}\left(\E_\Q[-X]-\alpha(\Q)\right), \quad X \in
    \Xc.\end{equation}
  Note that $\Pc_{eq}$ may be written in terms of $\mathcal{L}$ as
  \begin{equation}\label{eq:equivProbMeas} \Pc_{eq} = \left\lbrace
      \exp(f) \d \P \; : \; f \in \mathcal{L}
    \right\rbrace. \end{equation}
  Furthermore, using (iii) one obtains 
  \begin{equation}\label{eq:auxEq4} \Xc = \{ \bar{f} \circ
    \bar{I} \; : \; \bar{f} \in \Mc \}. \end{equation}

  Combining \eqref{eq:auxEq3}, \eqref{eq:equivProbMeas} and the
  definition of $\bar{\alpha}$ one obtains
  \[ \rho(X) = \sup_{f \in
      \mathcal{L}}\left(\E[-X\exp(f)]-\bar{\alpha}(f)\right), \] which
  can be rewritten as \eqref{eq:rhoRepresentation} by using
  \eqref{eq:auxEq4}.

  \textit{Step 2:} Note that one may also write \eqref{eq:rhoNNDef} as
  \begin{equation}\label{eq:rhoNNDef2}
    \rho^M(X)= \sup_{\substack{f \in \NN_{M,r(n+1),1} \\ \E[\exp(f\circ\bar{I})] = 1}}\left(\E[-X\exp(f\circ \bar{I})]-\bar{\alpha}(f\circ \bar{I})\right).
  \end{equation}
  Combining \eqref{eq:rhoNNDef2} with \eqref{eq:rhoRepresentation} and
  using
  $\NN_{M,r(n+1),1} \subset \NN_{M+1,r(n+1),1} \subset \mathcal{M}$,
  one obtains that $\rho^{M}(X) \leq \rho^{M+1}(X) \leq \rho(X)$ for
  all $M \in \N$. Thus it suffices to show that for any
  $\varepsilon > 0$ there exists $M \in \N$ such that
  $\rho^M(X) \geq \rho(X)-\varepsilon$.

  By \eqref{eq:rhoRepresentation}, for any $\varepsilon > 0$ one finds
  $\bar{f} \in \mathcal{M}$ such that
  \begin{align}
    \label{eq:auxEq5a} \E[\exp(\bar{f}\circ \bar{I})] & = 1, \\
    \label{eq:auxEq5b} \rho(X) - 2 \varepsilon & \leq \E[-X\exp(\bar{f}\circ \bar{I})]-\bar{\alpha}(\bar{f}\circ \bar{I}).
  \end{align}
  Precisely as in the proof of Proposition~\ref{lem:approx}, one may
  use Theorem~\ref{thm:universalApprox} to find
  $f^{(n)} \in \NN_{\infty,r(n+1),1}$ such that $\P$-a.s.,
  $f^{(n)}\circ \bar{I}$ converges to $\bar{f} \circ \bar{I}$ as
  $n \to \infty$. Combining this with \eqref{eq:auxEq5a}, one obtains
  that for all $n$ large enough,
  $ c_n:= \log(\E[\exp(f^{(n)} \circ \bar{I})]) $ is well-defined and
  that $\bar{f}^{(n)}\circ \bar{I}$ also converges $\P$-a.s. to
  $ \bar{f} \circ \bar{I}$, as $n \to \infty$, where
  $\bar{f}^{(n)}:= f^{(n)} - c_n$. Using this, \eqref{eq:auxEq5b} and
  assumption (ii), for some (in fact all) $n \in \N$ large enough one
  obtains
  \begin{equation}\label{eq:auxEq6} \rho(X) - \varepsilon \leq
    \E[-X\exp(\bar{f}^{(n)} \circ \bar{I})]-\bar{\alpha}(\bar{f}^{(n)}
    \circ \bar{I}). \end{equation}
  From $\NN_{\infty,r(n+1),1}-c_n = \NN_{\infty,r(n+1),1}$ and from the choice of $\NN_{M,r(n+1),1}$, one has $\bar{f}^{(n)}\in \NN_{M,r(n+1),1}$ for $M$ large enough. By combining this with \eqref{eq:auxEq6} and the choice of $c_n$ one obtains
  \[\rho(X)-\varepsilon \leq \rho^M(X), \]
  as desired.
\end{proof}

Combining \eqref{eq:objective1NN} and \eqref{eq:rhoNNDef}, one thus
approximates \eqref{eq:objective1} for $X=-Z$ by solving
\begin{equation} \label{eq:minimaxObj} \inf_{\theta_0 \in \Theta_{M}}
  \sup_{\theta_1 \in \Theta_{M,r(n+1),1}} J(\theta), \end{equation}
where $\theta = (\theta_0,\theta_1)$,
\[ J(\theta):= \E\left[-\mathrm{PL}(Z,0,\delta^{\theta_0})
    \exp(F^{\theta_1} \circ \bar{I})\right]
  -\bar{\alpha}(F^{\theta_1}\circ \bar{I}) - \lambda_0
  (\E[\exp(F^{\theta_1} \circ \bar{I})] -1 ) \] and $\lambda_0$ is a
Lagrange multiplier.

We conclude this section by arguing that the objective $J$ in
\eqref{eq:minimaxObj} indeed satisfies (ML1) and (ML2). This is
standard (c.f. Section~\ref{subsec:numericalSol}) for all terms in the
sum except for $\bar{\alpha}(F^{\theta_1}\circ \bar{I})$ and so we
only consider this term.

Recall that $\Omega$ is finite and consists of $N$ elements, thus
$\Xc=\{X \colon \Omega \to \R \}$ can be identified with
$\R^N$. As for standard backpropagation the compositional structure
can be used for efficient computation:

\begin{proposition}
  Suppose $\bar{\alpha}$ can be extended to
  $\bar{\alpha} \colon \Xc \to \R$ continuously
  differentiable, $\activF$ is continuously differentiable and
  $\NN_{M,r(n+1),1}$ is the set of neural networks with a fixed
  architecture (see Remark~\ref{rmk:NNex}). Then
  $J(\theta_1):=\bar{\alpha}(F^{\theta_1}\circ \bar{I})$,
  $\theta_1 \in \Theta_{M,r(n+1),1}$ is continuously differentiable
  and satisfies (ML1).
\end{proposition}

\begin{proof}
  Note that $F = F^{\theta_1}$ is parametrized by the matrices $A^\ell$
  and vectors $b^\ell$, $\ell=1,\ldots,L$, and that one may consider all
  partial derivatives separately.  Given
  $\bar{\alpha} \colon \Xc \to \R$ and
  $\nabla \bar{\alpha} \colon \Xc \to \Xc$, one thus
  aims at calculating
  $\partial_{A^{\ell}_{i,j}} \bar{\alpha}(F\circ \bar{I})$ and
  $\partial_{b_i^\ell} \bar{\alpha}(F\circ \bar{I})$ for
  $\ell = 1,\ldots, L, i = 1,\ldots,N_\ell, j = 1,\ldots,N_{\ell-1}$. This can
  be done by the chain rule: For $\theta \in \{A^{\ell}_{i,j}, b_i^\ell\}$,
  one has
  \[ \partial_\theta \bar{\alpha}(F\circ \bar{I}) = \sum_{m=1}^N
    \nabla \bar{\alpha}(F\circ \bar{I})(\omega_m)\partial_\theta
    F(\bar{I}(\omega_m)) \] and in particular (ML1) holds. \end{proof}

Furthermore, in the notation of the proof, for any $m = 1,\ldots N$
the derivative $\partial_\theta F(\bar{I}(\omega_m))$ can be
calculated using standard backpropagation algorithm (preceded by a
forward iteration) and so (ML2) holds as well. For the reader's
convenience we state it here: One sets $x^0=\bar{I}(\omega_m)$,
iteratively calculates $x^\ell:=F_\ell(x^{\ell-1})$ for $\ell=1,\ldots,L-1$ and
$x^L:=W_L(x^{L-1})$. Then (this is the backward pass) one sets
$J^L := A^L$ and calculates iteratively
$J^{\ell}=J^{\ell+1} dF_{\ell}(x^{\ell-1})$ for $\ell = L-1,\ldots,1$, where
\[
  dF_{\ell}(x^{\ell-1}) = \mathrm{diag}(\activF'(W_\ell x^{\ell-1})) A^\ell.\]
From this one may use again the chain rule to obtain for any
$\ell = 1,\ldots L, i = 1,\ldots,N_\ell, j = 1,\ldots,N_{\ell-1}$ the
derivatives of $F$ with respect to the parameters as
\[\begin{aligned} \partial_{A^{\ell}_{i,j}} F(\bar{I}(\omega_m)) & = J^{\ell+1}_i \activF'((W_\ell x^{\ell-1})_i) x^{\ell-1}_j \\
    \partial_{b_i^\ell} F(\bar{I}(\omega_m)) & = J^{\ell+1}_i
    \activF'((W_\ell x^{\ell-1})_i).
  \end{aligned}\]

\section{Numerical experiments and results}\label{sec:Heston}
After having introduced the optimal hedging problem
\eqref{eq:objective1} in Section~\ref{sec:Hedging} and described in
Section~\ref{sec:NeuralNets} how one may numerically approximate the
solution by \eqref{eq:objective1NN} using neural networks, we now turn
to numerical experiments to illustrate the feasibility of the
approach. We start by explaining in
Section~\ref{subsec:implementation} the modeling choices in
detail. The remainder of this section will then be devoted to
examining the following three questions:
\begin{itemize}
\item Section~\ref{subsec:benchmark}: How does neural network hedging
  (for different risk-preferences) compare to the benchmark in a
  Heston model without transaction costs?
\item Section~\ref{subsec:asymptotics}: What is the effect of
  proportional transaction costs on the exponential utility
  indifference price?
\item Section~\ref{subsec:highD}: Is the numerical method scalable to
  higher dimensions?
\end{itemize}

\subsection{Setting and Implementation}\label{subsec:implementation}
For the results presented here we have chosen a time horizon of $30$
trading days with daily rebalancing.  Thus, $T=30/365$, $n=30$ and the
trading dates are $t_i=i/365$, $i=0,\ldots,n$. As explained in
Section~\ref{sec:NeuralNets} and Remark~\ref{rmk:recurrent}, the
number of units $\delta_{t_i} \in \R^d$ that the agent decides to hold
in each of the instruments at $t_i$ is parametrized by a semi-recurrent
neural network: we set
$\delta_{k}^\theta = F^{\theta_k}(I_{k},\delta_{k-1}^\theta)$
where $F^{\theta_k}$ is a feed forward neural network with two hidden
layers and $I_{k}=\Phi(S_{0},\ldots,S_{k})$ for some
$\Phi\colon \R^{(k+1)d} \to \R^d$ specified below. More precisely, in
the notation of Definition~\ref{def:nn}, $F^{\theta_k}$ is a neural
network with $L=3$, $N_0=2d$, $N_1=N_2=d+15$, $N_3=d$ and the
activation function is always chosen as $\activF(x)=\max(x,0)$. The
weight matrices and biases are the parameters to be optimized in
\eqref{eq:objective1NN}. Note that these are different for each $k$.

Having made these choices, the algorithm outlined in
Section~\ref{sec:NeuralNets} can now be used for approximate 
 hedging \textit{in any market situation}: given sample
trajectories of the hedging instruments $S(\omega_m)$, samples of the
payoff $Z(\omega_m)$ and associated weights $\P[\{\omega_m \}]$ for
$m=1,\ldots,N$ (on a finite probability space
$\Omega=\{\omega_1,\ldots,\omega_N\}$), for any choice of transaction
cost structure $c$ and any risk measure $\rho$
one may now use the algorithm outlined in Section~\ref{sec:NeuralNets}
to calculate close-to optimal hedging strategies and approximate
minimal prices. Of course, for a path-dependent derivative with payoff
$Z=G(S_0,\ldots,S_T)$ with $G\colon (\R^{d})^{n+1} \to \R$ one obtains
samples of the payoff by simply evaluating $G$ on the sample
trajectories of $S$.

Different risk measures $\rho$, transaction cost functions
$c$ and payoffs $Z$ will be used in the examples
and so these are described separately in each of the subsequent
sections. To illustrate the feasibility of the algorithm and have a
benchmark at hand for comparison (at least in the absence of
transaction costs), we have chosen to generate the sample paths of $S$
from a standard stochastic volatility model under a risk-neutral
measure $\P$. Thus in most of the examples below, the process $S$
follows (a discretization of) a Heston model, see the beginning of
Section~\ref{subsec:benchmark} below.  But we stress again that, as
explained above, the algorithm is \textit{model independent} in the
sense that no information about the Heston model is used except for
the (weighted) samples of the price and variance process.

The algorithm has been implemented in Python, using Tensorflow to
build and train the neural networks. To allow for a larger learning
rate, the technique of batch normalization (see \cite{Ioffe2015} and
\cite[Chapter~8.7.1]{Goodfellow2016}) is used in each layer of each
network right before applying the activation function.  The network
parameters are initialized randomly (drawn from uniform and normal
distribution). For network training the Adam algorithm (see
\cite{Kingma2015}, \cite[Chapter~8.5.3]{Goodfellow2016}) with a
learning rate of $0.005$ and a batch size of $256$ has been
used. Finally, the model hedge for the benchmark in
Section~\ref{subsec:benchmark} has been calculated using Quantlib.

\begin{remark} For the numerical experiments in this article the optimality criteria in \eqref{eq:Jdef} and \eqref{eq:JEntropic} are specified under a risk-neutral measure. Thus, an optimal hedging strategy is based on market anticipations of future prices. Alternatively, one could use a statistical measure. The algorithm presented here can be applied also in this case.
\end{remark}

\subsection{Benchmark: No transaction costs} \label{subsec:benchmark}
As a first example, we consider hedging without transaction costs in a
Heston model. In this example the risk measure $\rho$ is chosen as the
average value at risk (also called conditional value at risk or
expected shortfall), defined for any random variable $X$ by
\begin{equation} \rho(X):=
  \label{eq:AVar} \frac{1}{1-\alpha} \int_0^{1-\alpha} \mathrm{VaR}_{\gamma}(X) \d \gamma
\end{equation}
for some $\alpha \in [0,1)$, where
$\mathrm{VaR}_{\gamma}(X):=\inf \{m \in \R \, : \P(X<-m) \leq \gamma
\}$. An alternative representation of $\rho$ of type \eqref{eq:OCE} is
discussed in Example~\ref{ex:CVAR}. We refer to
\cite[Section~4.4]{Foellmer2016} for further details. Note that
different levels of $\alpha$ correspond to different levels of
risk-aversion, ranging from risk-neutral for $\alpha$ close to $0$ to
very risk-averse for $\alpha$ close to $1$. The limiting cases are
$\rho(X)=-\E[X]$ for $\alpha=0$ and
$\lim_{\alpha \uparrow 1} \rho(X)= -\mathrm{essinf}(X)$, see
\cite[p.234 and Remark~4.50]{Foellmer2016}.

\subsubsection*{A brief reminder on the Heston model}
Recall that a Heston model is specified by the stochastic differential
equations
\begin{equation}\label{eq:HestonModel}\begin{aligned} \d S^{1}_t & = \sqrt{V_t} S^{1}_t \d B_t, \quad \text{ for } t > 0 \text{ and } S^{1}_0 = s_0 \\
    \d V_t & = \alpha (b - V_t) \d t + \sigma \sqrt{V_t} \d W_t ,
    \quad \text{ for } t > 0 \text { and } V_0 =
    v_0, \end{aligned} \end{equation} where $B$ and $W$ are
one-dimensional Brownian motions (under a probability measure $\Q$)
with correlation $\rho \in [-1,1]$ and $\alpha$, $b$, $\sigma$, $v_0$
and $s_0$ are positive constants. Below we have chosen $\alpha=1$,
$b=0.04$, $\rho=-0.7$, $\sigma=2$, $v_0=0.04$ and $s_0=100$,
reflecting a typical situation in an equity market.

Here $S^{1}$ is the price of a liquidly tradeable asset and $V$ is
the (stochastic) variance process of $S^{1}$, modeled by a
Cox-Ingersoll-Ross (CIR) process. $V$ itself is not tradable directly,
but only through options on variance. In our framework this is modeled
by an idealized variance swap with maturity $T$, i.e. we set
$\Fc^{H}_t:=\sigma((S^{1}_s,V_s) \, : \, s \in [0,t])$ and
\begin{equation}\label{eq:varSwapPrice} S^{2}_t := \E_\Q
  \left[\left. \int_0^T V_s \dd s \right|\Fc^{H}_t\right], \quad t \in
  [0,T], \end{equation}
and consider $(S^{1},S^{2})$ as the prices of liquidly tradeable assets. 
A standard calculation\footnote{For example, one may use that $(\log(S^{1}),V)$ 
is an affine process to see that the conditional expectation in \eqref{eq:varSwapPrice}
 can be taken only with respect to $\sigma(V_t, s \in [0,t])$. 
 This conditional expectation can then be calculated by using the 
 SDE for $V$ or by directly inserting the expression from e.g. \cite[Section~3]{dufresne2001}.}
  shows that \eqref{eq:varSwapPrice} is given as
\begin{equation}\label{eq:varSwapPrice2} S^{2}_t = \int_0^t V_s \dd
  s + L(t,V_t) \end{equation}
where 
\[L(t,v)= \frac{v-b}{\alpha}(1-e^{-\alpha(T-t)})+b(T-t).
\]

Consider now a European option with payoff $g(S^{1}_T)$ at $T$ for
some $g \colon \R \to \R$. Its price (under $\Q$) at $t \in [0,T]$ is
given as $H_t:=\E_\Q[g(S^{1}_T)| \Fc^H_t]$. By the Markov property
of $(S^{1},V)$, one may write the option price at $t$ as
$H_t = u(t,S^{1}_t,V_t)$ for some
$u \colon [0,T]\times[0,\infty)^2 \to \R$. Assuming that $u$ is
sufficiently smooth, one may apply It\^{o}'s formula to $H$ and use
\eqref{eq:varSwapPrice2} to obtain
\begin{equation} \label{eq:replicationHeston} 
g(S^{1}_T)=q+ \int_0^T
  \delta^{1}_t \dd S^{1}_t + \int_0^T \delta^{2}_t \dd
  S^{2}_t \end{equation} 
where $q = \E_\Q[g(S^{1}_T)]$ and
\begin{equation}\label{eq:modelDeltaHeston}
  \delta^{1}_t:=\partial_s u(t,S^{1}_t,V_t) \text{ and } \delta^{2}_t:= \frac{\partial_v 
  u(t,S^{1}_t,V_t)}{\partial_v L(t,V_t)}.
\end{equation}
Thus, if continuous-time trading was possible,
\eqref{eq:replicationHeston} shows that the option payoff can be
replicated perfectly by trading in $(S^{1},S^{2})$ according to
the strategy \eqref{eq:modelDeltaHeston}.

\begin{remark}
  The strategy \eqref{eq:modelDeltaHeston} depends on $V_t$. Although
  not observable directly, an estimate can be obtained by estimating
  $\int_0^t V_s \dd s$ and solving \eqref{eq:varSwapPrice2} for $V_t$.
\end{remark}

\subsubsection*{Setting: Discretized Heston model}

In addition to the setting explained in detail in
Section~\ref{subsec:implementation}, here we set $d=2$, consider no
transaction costs (i.e. $C_T\equiv 0$) and generate sample trajectories
of the price process of the hedging instruments from a discretely
sampled Heston model. Thus, $S=(S_{0},\ldots,S_{n})$ and for any
$k=0,\ldots,n$, $S_{k} = (S_{k}^{1},S_{k}^{2})$ is given by
\eqref{eq:HestonModel} and \eqref{eq:varSwapPrice2} under $\Q$. The
sample paths of $S$ are generated by (exact) sampling from the
transition density of the CIR process (see
\cite[Section~3.4]{Glasserman2004}) and then using the (simplified)
Brodie-Kaya scheme (see \cite{Andersen2010} and
\cite{Broadie2006}).\footnote{This corresponds to replacing $V$ in the
  SDE for $S^{1}$ in \eqref{eq:HestonModel} by a piecewise constant
  process and the integral in \eqref{eq:varSwapPrice2} by a sum.}
Generating independent samples of $S$ according to this scheme can now
be viewed as sampling from a uniform distribution on a (huge) finite
probability space $\Omega$.\footnote{To be more precise, one replaces
  the normal distributions appearing in the simulation scheme for $S$
  by (arbitrarily fine) discrete distributions.} Thus, in the notation
of Section~\ref{subsec:implementation} one has $\P[\{\omega_m\}]=1/N$ for
all $m=1,\ldots,N$ with each $S(\omega_m)$ corresponding to a sample
of the Heston model generated as explained above.

If continuous-time trading was possible, any European option could be
replicated perfectly by following the strategy
\eqref{eq:modelDeltaHeston}. However, in the present setup the hedging
portfolio can only be adjusted at discrete time-points. Nevertheless
one may choose
$\delta_{k}^H:= (\delta^{1}_{k},\delta^{2}_{k})$ for
$k = 0\ldots n-1$ with $\delta^{1}, \delta^{2}$ defined by
\eqref{eq:modelDeltaHeston} and charge the risk-neutral price
$q$. This will be referred to as the model-delta hedging strategy (or
simply model hedge) and serves as a benchmark.

Finally, in order to compare the neural network strategies to this
benchmark, the network input is chosen as
$I_{k}=(\log(S_{k}^{1}),V_{k})$. One could also replace
$V_{k}$ by $S^{2}_{k}$ instead. The network structure at
time-step $t_k$ is illustrated in Figure~\ref{Fig0}.

\begin{figure}
  \centering
  \includegraphics[trim= 0mm 70mm 0mm 45mm,
  clip,width=0.8\textwidth]{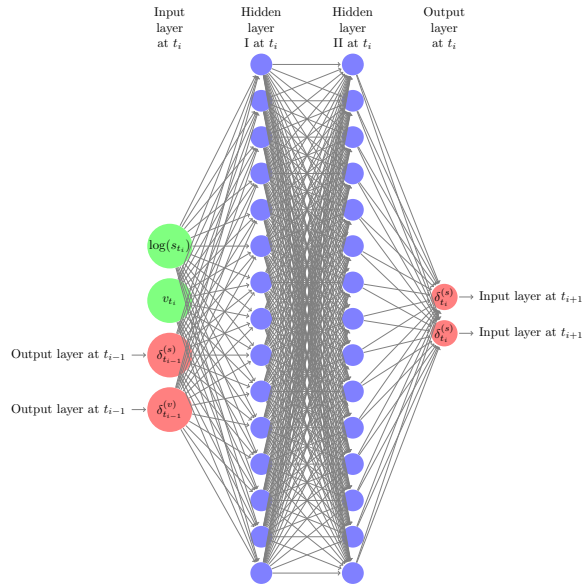}
  \caption{Recurrent network structure}
  \label{Fig0}
\end{figure}

\subsubsection*{Results}
We now compare the model hedge $\delta^H$ to the deep hedging
strategies $\delta^\theta$ corresponding to different
risk-preferences, captured by different levels of $\alpha$ in the
average value at risk \eqref{eq:AVar}.

As a first example, consider a European call option,
i.e. $Z=(S^{1}_T-K)^+$ with $K=s_0$. Following the methodology
outlined in Section~\ref{subsec:implementation}, we calculate a
(close-to) optimal parameter $\theta$ for \eqref{eq:objective1NN} with $X=-Z$ and
denote by $\delta^\theta$ and $p^\theta_0$ the (close-to) optimal
hedging strategy and value of \eqref{eq:objective1NN},
respectively. By definition of the indifference price
\eqref{eq:priceDef}, the approximation property
Proposition~\ref{lem:approx}, Proposition~\ref{lem:martingale} and
$\rho(0)=0$, $p^\theta_0$ is an approximation to the indifference price
$p(Z)$.  As an out-of-sample test, one can then simulate another set
of sample trajectories (here $10^6$) and evaluate the terminal hedging
errors $q - Z+(\delta^H \cdot S)_T$ (model hedge) and
$p^\theta_0 - Z+(\delta^{\theta} \cdot S)_T$ (CVar) on each of them. In
fact, since the risk-adjusted price $p^\theta_0$ is higher than the
risk-neutral price $q=1.69$ (as shown in
Proposition~\ref{lem:martingale}(ii)), for (CVar) we have evaluated
$q - Z+(\delta^{\theta} \cdot S)_T$, i.e. the hedging error from using
the optimal strategy associated to $\rho$, but only charging the
risk-neutral price $q$. This is shown in a histogram in
Figure~\ref{Fig1} for $\alpha=0.5$, yielding a risk-adjusted price
$p^\theta_0=1.94$.  As one can see, the hedging performance of
$\delta^H$ and $\delta^\theta$ is very similar. In particular
\begin{itemize}
\item for this choice of risk-preferences ($\rho$ as in
  \eqref{eq:AVar} with $\alpha = 0.5$) the optimal strategy in
  \eqref{eq:objective1} is close to the model hedge $\delta^H$,
\item the neural network strategy $\delta^\theta$ is able to
  approximate  well the optimal strategy in \eqref{eq:objective1}.
\end{itemize}
This is also illustrated by Figure \ref{Fig3}, where the strategies
$\delta_t^\theta$ and $\delta_t^H$ at a fixed time-point $t$ are
plotted conditional on $(S_t^{1},V_t) = (s,v)$ on a grid of values
for $(s,v)$. To make this last comparison fully sensible instead of
the recurrent network structure
$\delta_{k}^\theta = F^{\theta_k}(I_{k},\delta_{k-1}^\theta)$
here a simpler structure $\delta_{k}^\theta = F^{\theta_k}(I_{k})$
is used. The hedging performance for this simpler structure is,
however, very similar, see Figure~\ref{Fig5}. Of course, this is also
expected from \eqref{eq:modelDeltaHeston}.\footnote{For non-zero
  transaction costs this is not true anymore, i.e. the recurrent
  network structure is needed. For example, Figure~\ref{Fig6} is
  generated for precisely the same parameters as Figure~\ref{Fig5},
  except that $\alpha = 0.99$ and proportional transaction costs are
  incurred, i.e. \eqref{eq:propTransCosts} with
  $\varepsilon = 0.01$.}

\begin{figure}
  \centering
  \includegraphics[width=0.9\textwidth]{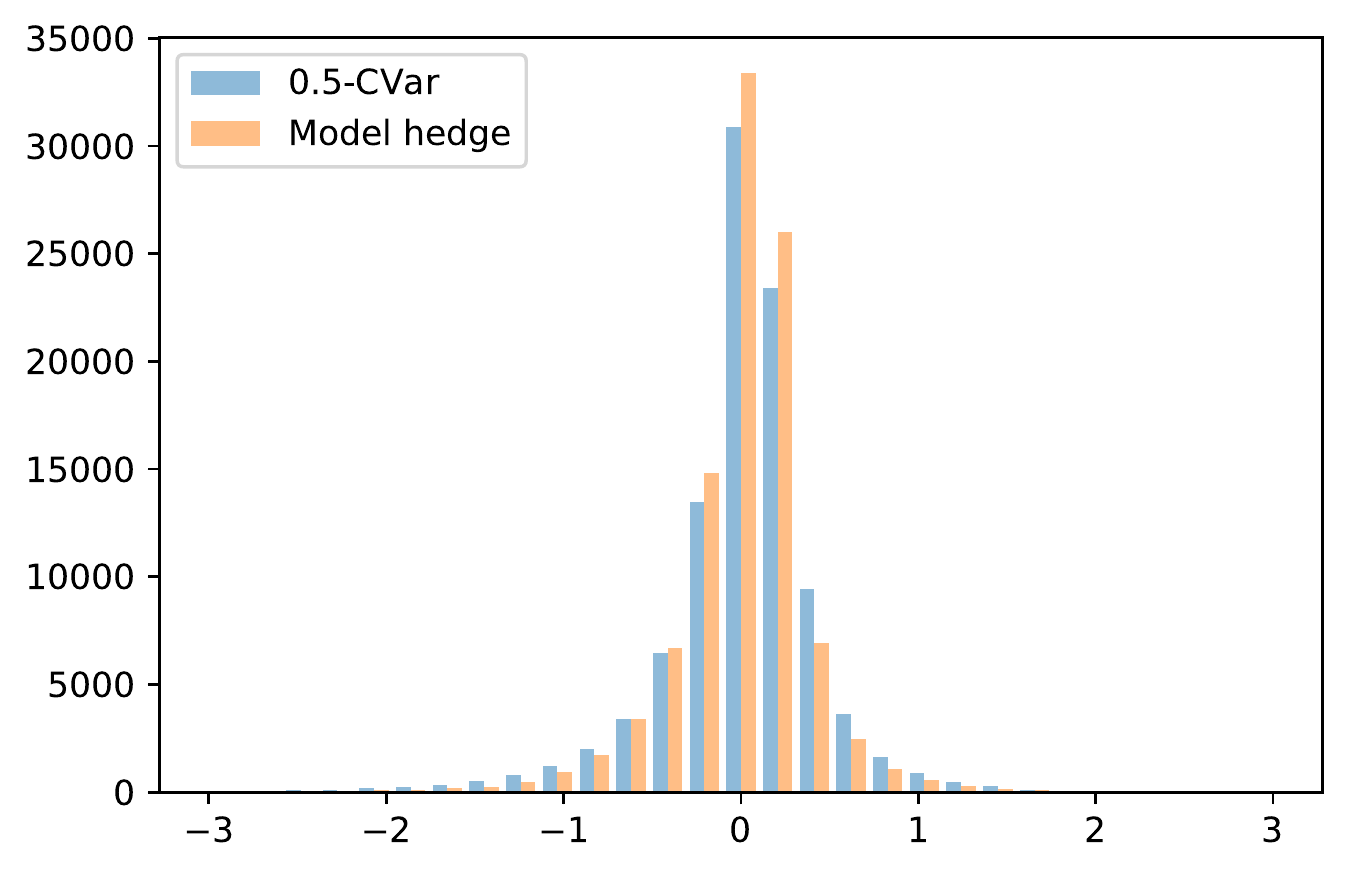}
  \caption{Comparison of model hedge and deep hedge associated to
    $50\%$-expected shortfall criterion.}
  \label{Fig1}
\end{figure}

\begin{figure}
  \centering
  \includegraphics[width=0.85\textwidth]{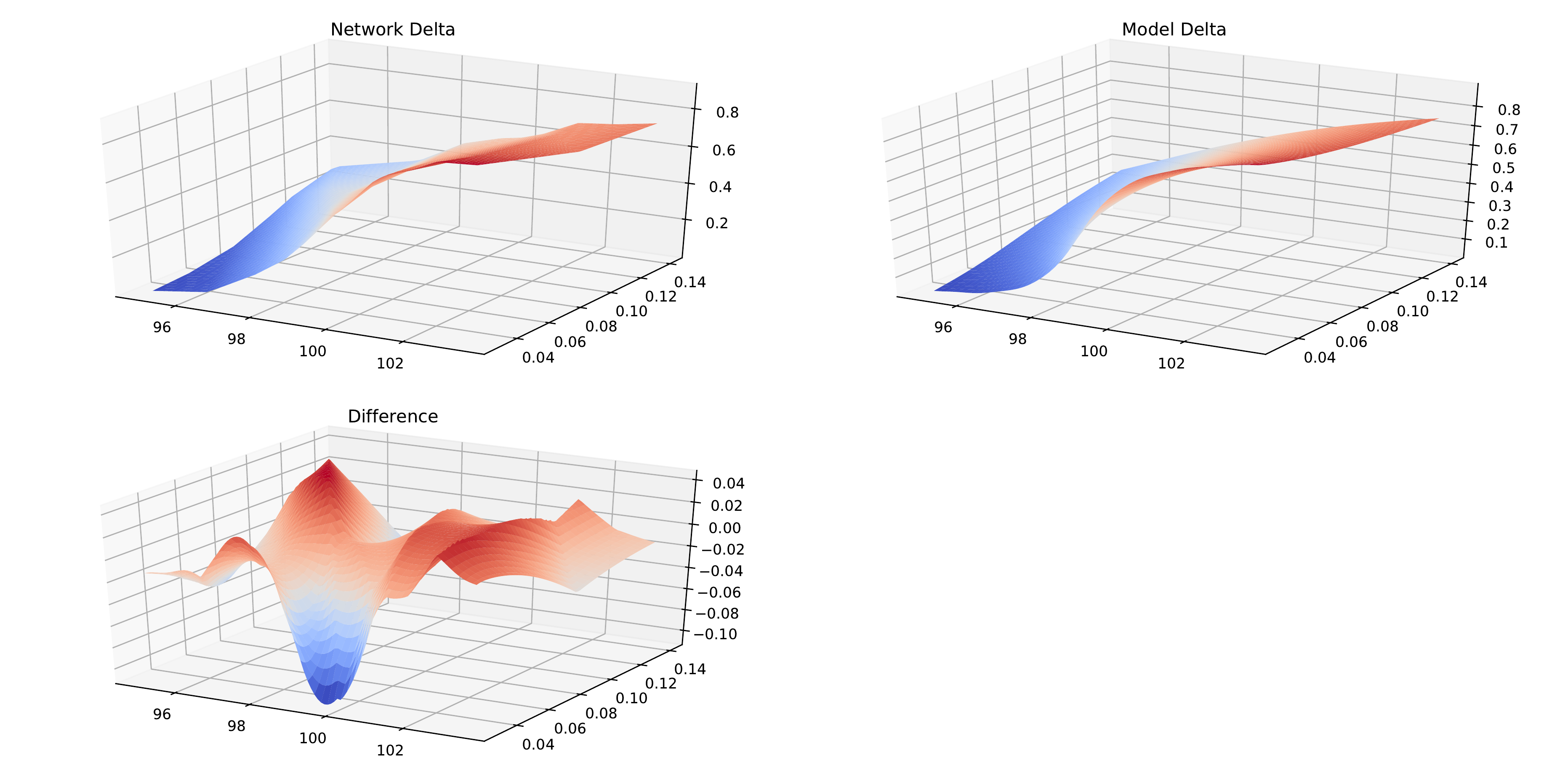}
  \caption{$\delta^{H,(1)}_t$ and neural network approximation as a
    function of $(s_t,v_t)$ for $t=15$ days}
  \label{Fig3}
\end{figure}

\begin{figure}
  \centering
  \includegraphics[width=0.85\textwidth]{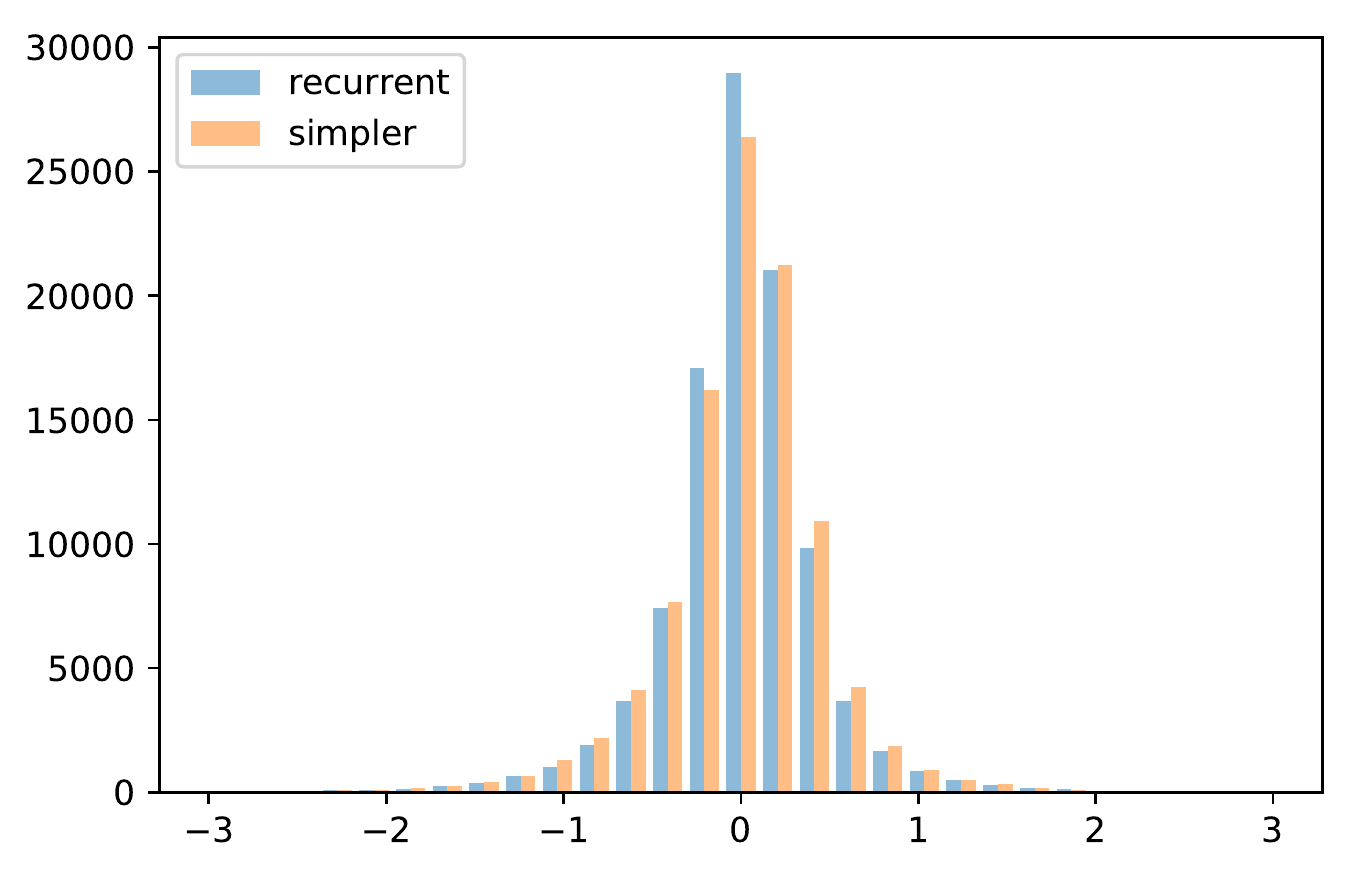}
  \caption{Comparison of recurrent and simpler network structure (no
    transaction costs).}
  \label{Fig5}
\end{figure}

\begin{figure}
  \centering
  \includegraphics[width=0.85\textwidth]{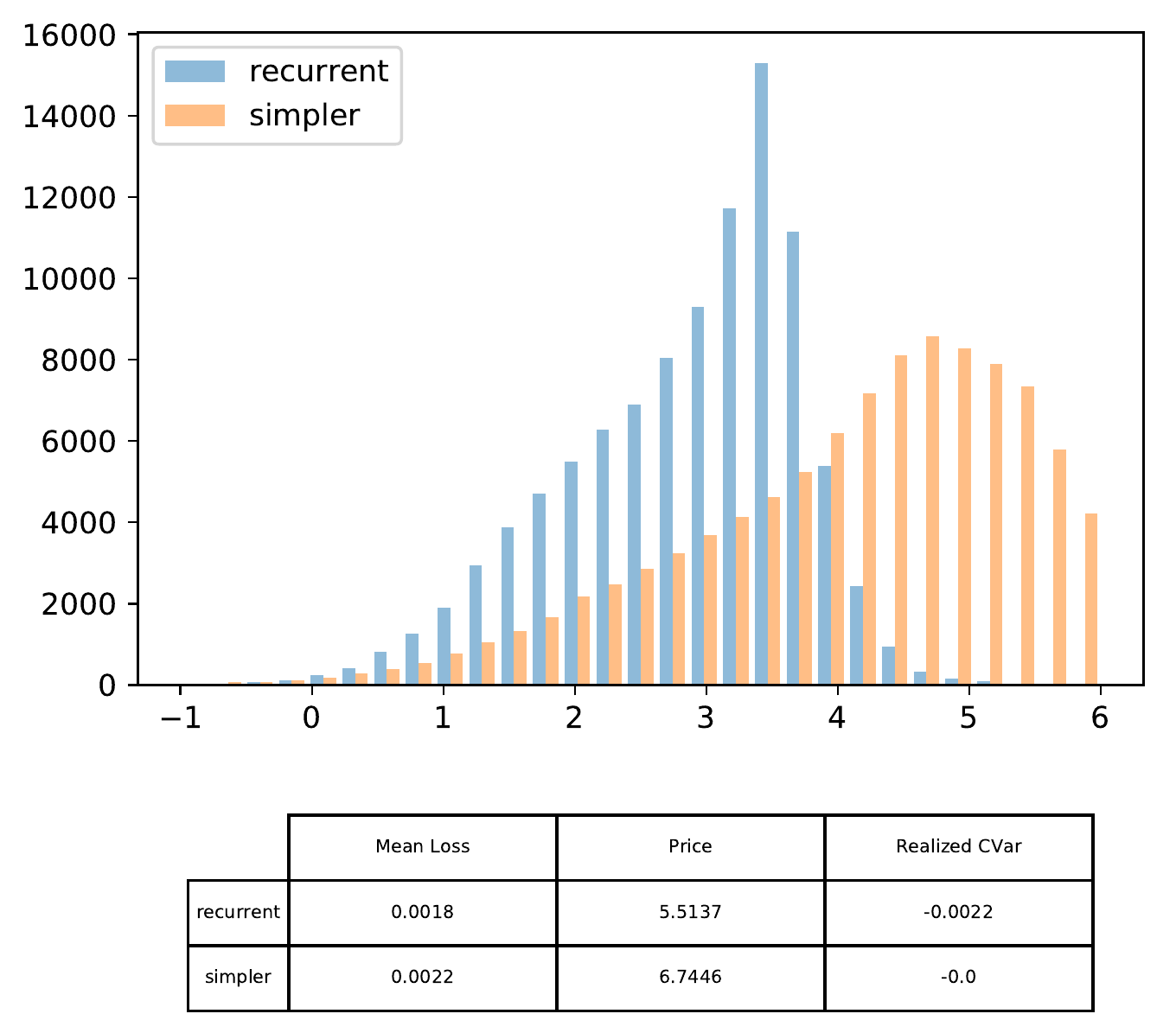}
  \caption{Network architecture matters: Comparison of recurrent and
    simpler network structure (with transaction costs and $99\%$-CVar
    criterion).}
  \label{Fig6}
\end{figure}

\begin{figure}
  \centering
  \includegraphics[width=0.75\textwidth]{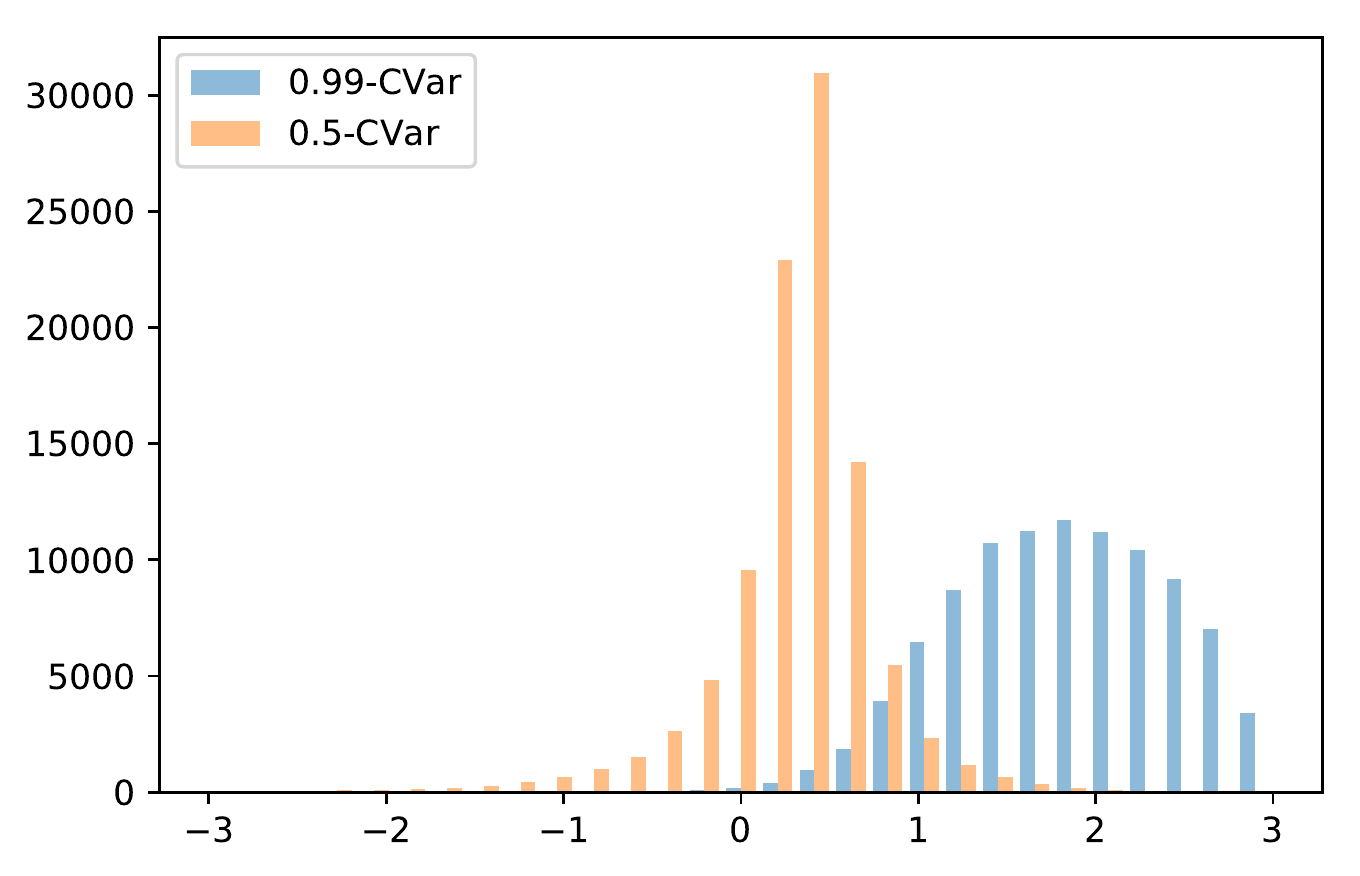}
  \caption{Comparison of 99\%-CVar and 50\%-CVar optimiality
    criterion.}
  \label{Fig31}
\end{figure}

\begin{figure}
  \centering
  \includegraphics[width=0.75\textwidth]{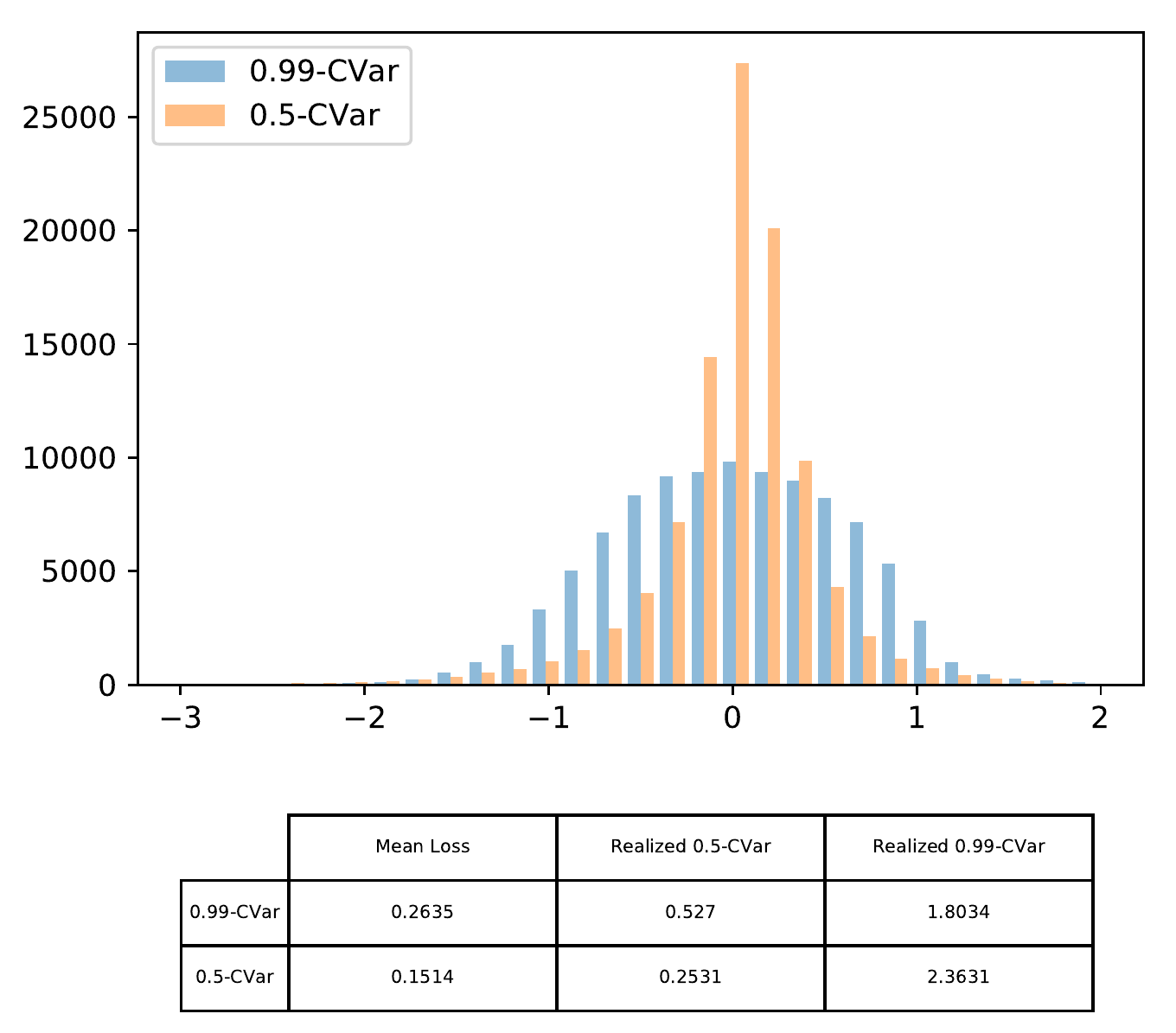}
  \caption{Comparison of 99\%-CVar and 50\%-CVar optimiality
    criterion, normalized to risk-neutral price.}
  \label{Fig32}
\end{figure}

A more extreme case is shown in Figure~\ref{Fig31}, where instead of
the model hedge the $99\%$-CVar criterion is used,
i.e. $\alpha=0.99$. This results in a significantly higher
risk-adjusted price $p^\theta_0=3.49$. If both the $50\%$ and
$99\%$-CVar optimal strategies are used, but only the risk-neutral
price is charged (see Figure~\ref{Fig32}) one can clearly see the risk
preferences: the $50\%$-CVar strategy is more centered at $0$ and also
has a smaller mean hedging error, but the $99\%$-expected shortfall
strategy yields smaller extreme losses (c.f. also the realized
$99\%$-CVar \textit{loss} value realized on the test sample, shown in
the table below Figure~\ref{Fig32}).

To further illustrate the implications of risk-preferences on hedging,
as a last example we consider selling a call-spread,
i.e. $Z=[(S^{1}_T-K_1)^+-(S^{1}_T-K_2)^+]/(K_2-K_1)$ for
$K_1 < K_2$. Here we have chosen $K_1=s_0$, $K_2 = 101$. Proceeding as
above, we compare the model hedge to the more risk-averse hedging
strategies associated to $\alpha=0.95$ and $\alpha=0.99$. The
strategies (on a grid of values for spot and variance) are shown in
Figures~\ref{FigSpread1} and \ref{FigSpread2}. The model hedge would
again correspond to $\alpha=0.5$. As one can see for higher levels of
risk-aversion, the strategy flattens. From a practical perspective,
this precisely corresponds to a barrier shift, i.e. a more risk-averse
hedge for a call spread with strikes $K_1$ and $K_2$ actually aims at
hedging a spread with strikes $\tilde{K}_1$ and $K_2$ for
$\tilde{K}_1 < K_1$.

\begin{figure}
  \centering
  \includegraphics[width=0.85\textwidth]{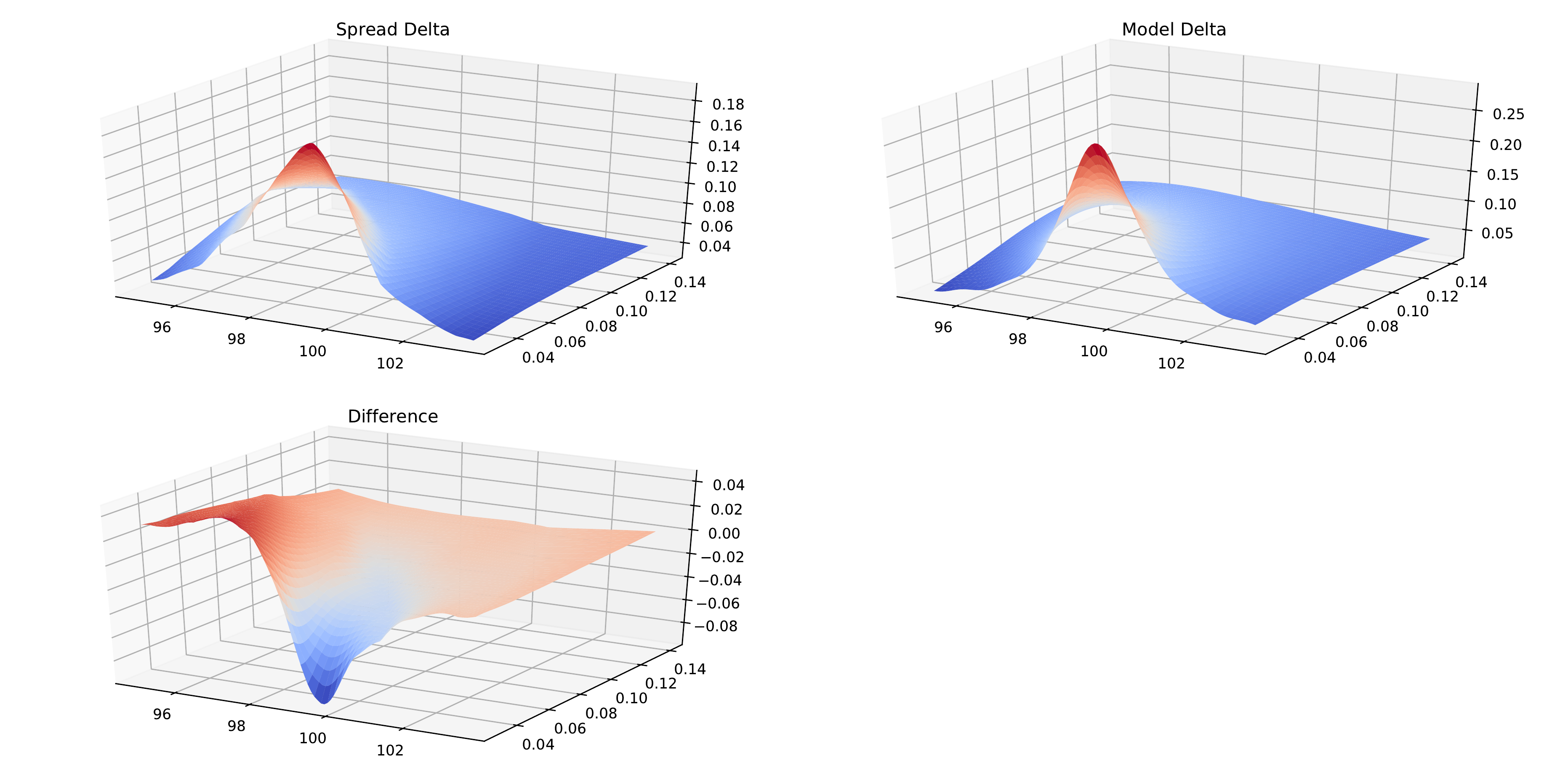}
  \caption{Call spread $\delta^{H,(1)}_t$ and neural network
    approximation as a function of $(s_t,v_t)$ for $t=15$ days}
  \label{FigSpread1}
\end{figure}

\begin{figure}
  \centering
  \includegraphics[width=0.85\textwidth]{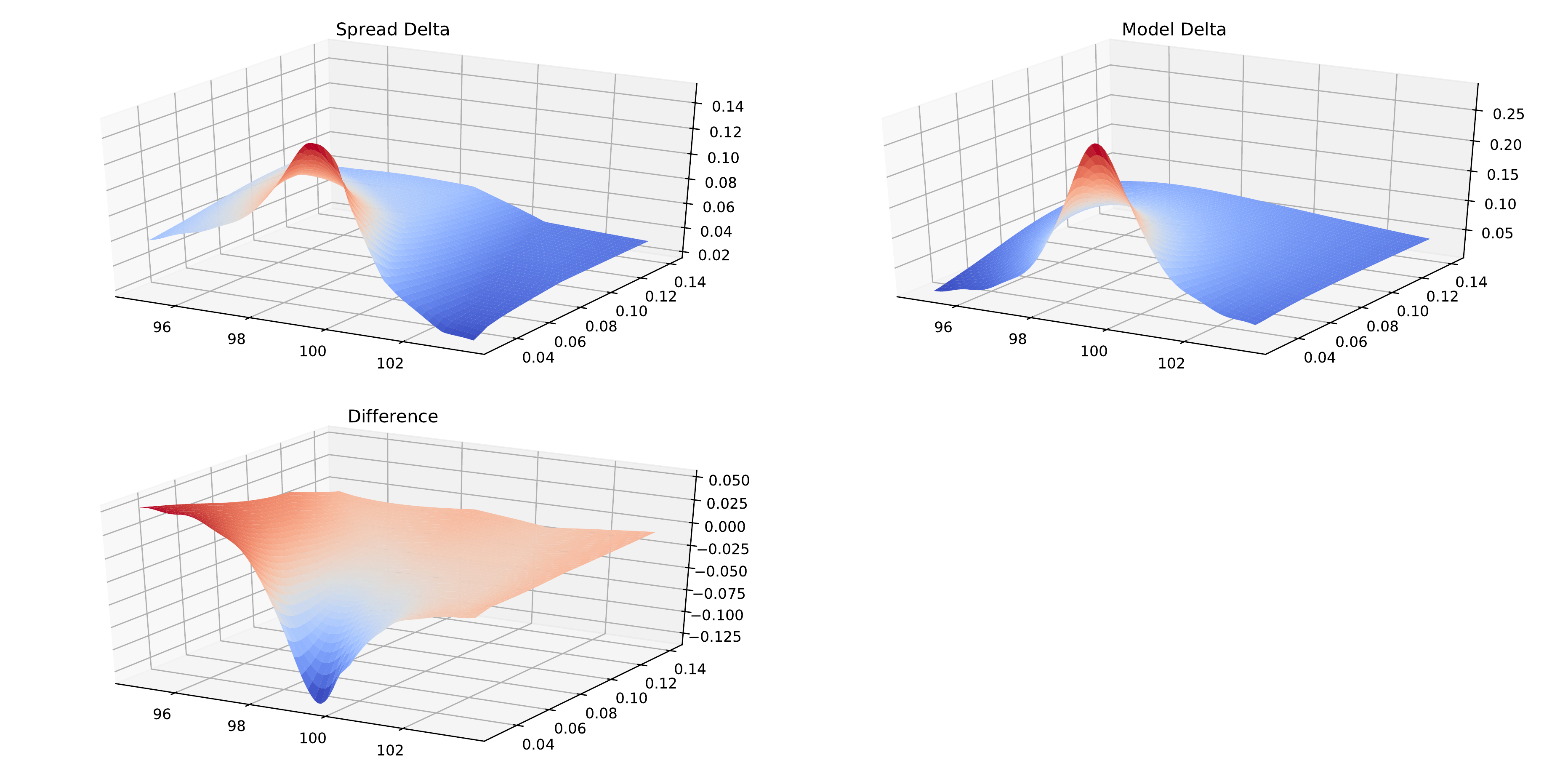}
  \caption{Call spread $\delta^{H,(1)}_t$ and neural network
    approximation as a function of $(s_t,v_t)$ for $t=15$ days}
  \label{FigSpread2}
\end{figure}

\subsection{Price asymptotics under proportional transaction
  costs} \label{subsec:asymptotics} In Section~\ref{subsec:benchmark}
we have seen that in a market without transaction costs, deep hedging
is able to recover the model hedge and can be used to calculate
risk-adjusted optimal hedging strategies.

The goal of this section is to illustrate the power of the methodology
by numerically calculating the indifference price \eqref{eq:priceDef}
in a multi-asset market with transaction costs.

So far, this has been regarded a highly challenging problem, see
e.g. the introduction of \cite{Kallsen2015}. For example, calculating
the exponential utility indifference price for a call option in a
Black-Scholes model involves solving a multidimensional nonlinear free
boundary problem, see e.g. \cite{Hodges1989},
\cite{Davis1993}. Motivated by this \cite{Whalley1997} have studied
asymptotically optimal strategies and price asymptotics for small
proportional transaction costs, i.e. for 
\begin{equation} \label{eq:propTransCosts} c_k(\mathrm{n})=\sum_{i=1}^d \varepsilon |\mathrm{n}^{i}| S^{i}_k
\end{equation}
and as $\varepsilon \downarrow 0$. One of the results in the
asymptotic analysis is that
\begin{equation}\label{eq:priceAsymptotics} p_\varepsilon - p_0 =
  O(\varepsilon^{2/3}),\quad \text{ as } \varepsilon \downarrow
  0, \end{equation}
where $p_\varepsilon=p_\varepsilon(Z)$ is the utility indifference price of $Z$ associated to transaction costs of size $\varepsilon$. In fact \eqref{eq:priceAsymptotics} is true in more general one-dimensional models, see \cite{Kallsen2015}, and the rate $2/3$ also emerges in a variety of related problems with proportional transaction costs, see e.g. \cite{Rogers2004}, \cite{MuhleKarbe2017} and the references therein.

Here we numerically verify \eqref{eq:priceAsymptotics} using the deep
hedging algorithm, first for a Black-Scholes model (for which
\eqref{eq:priceAsymptotics} is known to hold) and then for a Heston
model (with $d=2$ hedging instruments). For this latter case (or any
other model with $d>1$) there have been neither numerical nor
theoretical results on \eqref{eq:priceAsymptotics} previously in the
literature.

\subsubsection*{Black-Scholes model}
Consider first $d=1$ and $S_t=s_0 \exp(-t \sigma^2/2 + \sigma W_t)$,
where $\sigma > 0$ and $W$ is a one-dimensional Brownian motion. We
choose $\sigma = 0.2$, $s_0 = 100$ and use the explicit form of $S$ to
generate sample trajectories. Setting $I_{k} = \log(S_{k})$ and
proceeding precisely as in the Heston case (see
Sections~\ref{subsec:implementation} and \ref{subsec:benchmark}), we
may use the deep hedging algorithm to calculate the exponential
utility indifference price $p_{\varepsilon}$ for different values of
$\varepsilon$. Recall that we choose proportional transaction costs \eqref{eq:propTransCosts} and $\rho$ is the entropic risk measure
\eqref{eq:entropic} (see Lemma~\ref{lem:expIP}). For the numerical
example we take $\lambda=1$ and $Z=(S_T-K)^+$ with $K=s_0$ and we
calculate $p_{\varepsilon_i}$ for $\varepsilon_i=2^{-{i+5}}$,
$i=1,\dots,5$.

Figure~\ref{FigAsymptoticsBS} shows the pairs
$(\log(\varepsilon_i),\log(p_{\varepsilon_i}-p_0))$ (in red) and the
closest (in squared distance) straight line with slope $2/3$ (in
blue). Thus, in this range of $\varepsilon$ the relation
$\log(p_\varepsilon-p_0)= 2/3 \log(\varepsilon) + C$ for some
$C \in \R$ indeed holds true and hence also
\eqref{eq:priceAsymptotics}.

Note that trading is only possible at discrete time-points and so the
indifference price and the risk-neutral price do not coincide. Since
\eqref{eq:priceAsymptotics} is a result for continuous-time trading
(where $q=p_0$), we have compared to the risk-neutral price $q$ here
(thus neglecting the discrete-time friction in $p_\varepsilon$ for
$\varepsilon > 0$).

\subsubsection*{Heston model}
We now consider a Heston model with two
hedging instruments, i.e. $d=2$ and the setting is precisely as in
Section~\ref{subsec:benchmark}, except that here $\rho$ is chosen as
\eqref{eq:entropic} and proportional transaction costs \eqref{eq:propTransCosts} are incurred. Choosing
$\lambda=1$, $Z=(S_T^{1}-K)^+$ and $\varepsilon_i$ as in the
Black-Scholes case above, one can again calculate the exponential
utility indifference prices and show the difference to $p_0$ in a
log-log plot (see above) in a graph. These are shown as red dots in
Figure~\ref{FigAsymptoticsHeston}. Here the blue line in
Figure~\ref{FigAsymptoticsHeston} is the regression line, i.e. the
least squares fit of the red dots. The rate is very close to $2/3$ and
so it appears that the relation \eqref{eq:priceAsymptotics} also holds
in this case.

\begin{figure}
  \centering
  \includegraphics[width=0.85\textwidth]{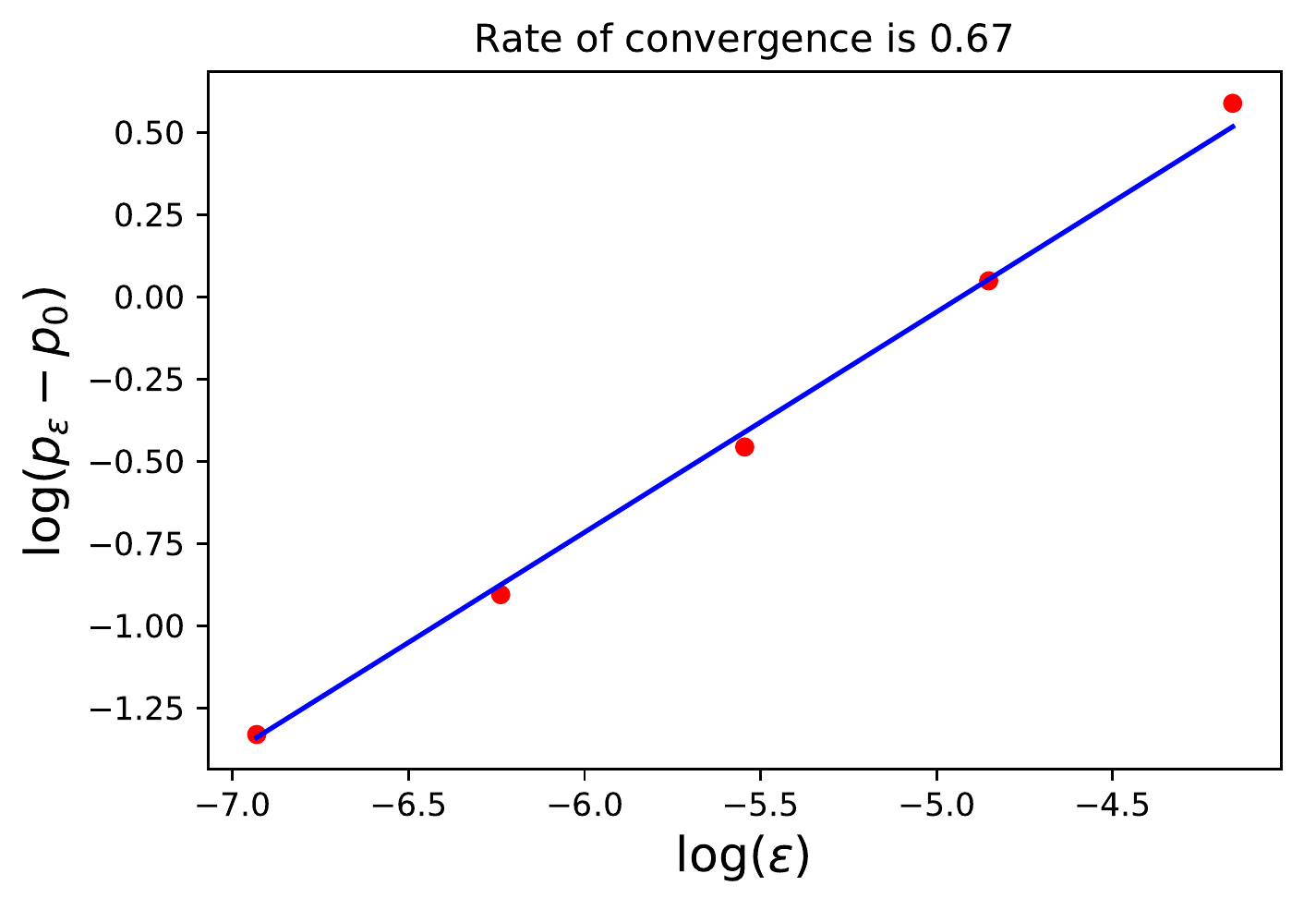}
  \caption{Black-Scholes model price asymptotics.}
  \label{FigAsymptoticsBS}
\end{figure}

\begin{figure}
  \centering
  \includegraphics[width=0.85\textwidth]{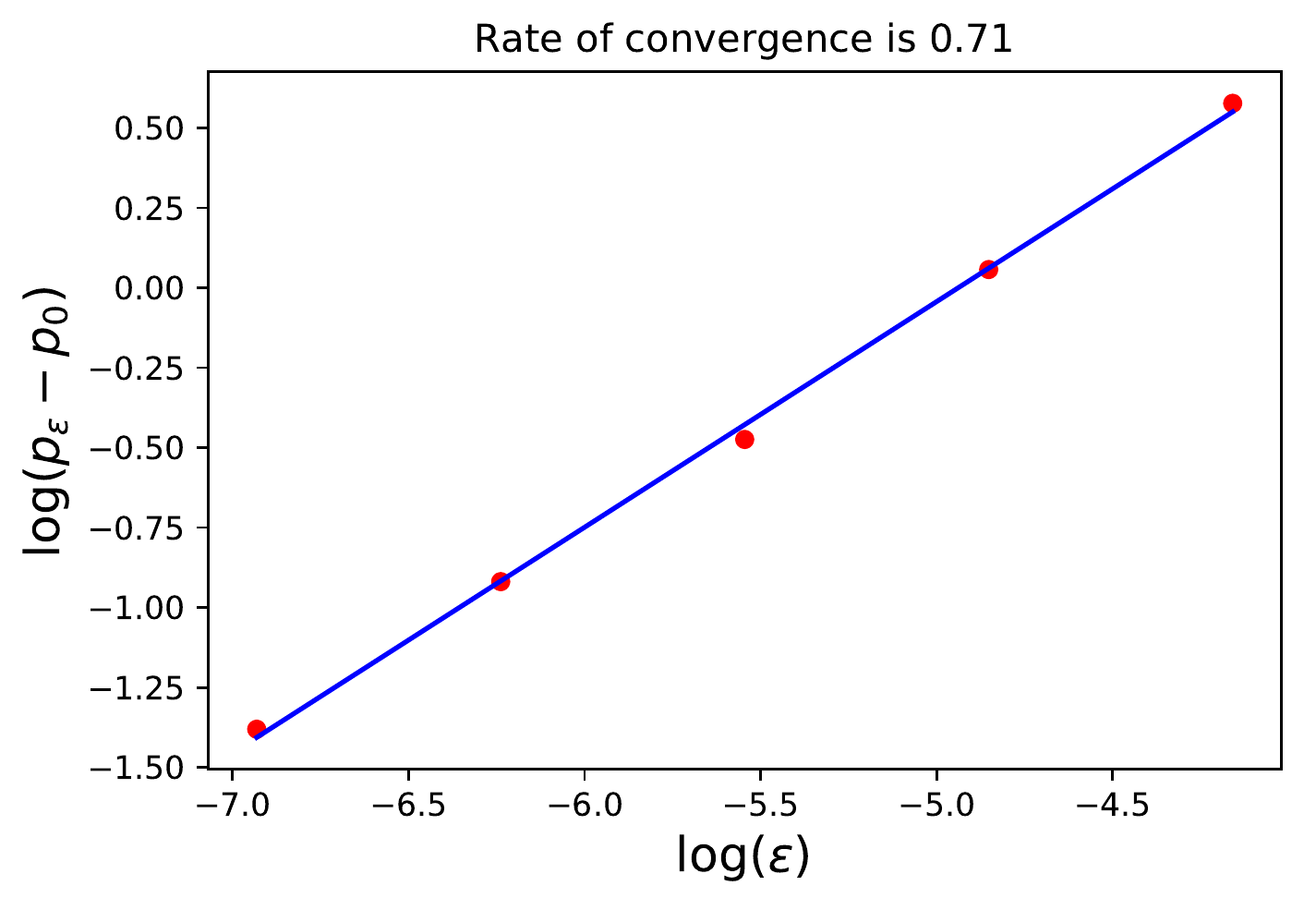}
  \caption{Heston model price asymptotics}
  \label{FigAsymptoticsHeston}
\end{figure}

\subsection{High-dimensional example} \label{subsec:highD} As a last
example consider a model built from $5$ separate Heston models,
i.e. $d=10$ and $(S^{h},S^{h+1})$ is the price process of spot and
variance swap in a Heston model (specified by \eqref{eq:HestonModel}
and \eqref{eq:varSwapPrice2}) for $h=1,\ldots,5$. To have a benchmark
at hand the $5$ models are assumed independent and each of them has
parameters as specified in Section~\ref{subsec:benchmark}. This choice
is of course no restriction for the algorithm and is only made for
convenience. The payoff is a sum of call options on each of the
underlyings, i.e. $Z = \sum_{h=1}^5 Z_h$ with
$Z_h = (S_T^{2h-1}-K)^+$ and $K=s_0=100$.  In a market with
continuous-time trading and no transaction costs, $Z$ can be
replicated perfectly by trading according to strategy
\eqref{eq:modelDeltaHeston} in each of the models. In particular, this
strategy is decoupled, i.e. the optimal holdings in
$(S^{h},S^{h+1})$ only depend on $(S^{(h)},S^{(h+1)})$. While in
the present setup trading is only possible at discrete time steps and
so the strategy optimizing \eqref{eq:objective1}, where $X=-Z$, leads to a
non-deterministic terminal hedging error \eqref{eq:terminalPL}, by
independence one still expects that the optimal strategy is decoupled
as above, at least for certain classes of risk measures. To see this
most prominently, here we consider \textit{variance optimal hedging}:
the objective is chosen as \eqref{eq:objective2} for $\ell(x)=x^2$ and
$p_0 = 5 q$, where $q = \E[Z_1]$.

Let $\delta \in \H$ and write
$\delta^{(2h-1:2h)}:=(\delta^{2h-1},\delta^{2h})$ for
$h=1,\ldots,5$ (and analogously for $S$).  If $\delta$ is decoupled,
i.e. such that $\delta^{(2h-1:2h)}$ is independent of $S^{(2j-1:2j)}$
for $j \neq h$, then by independence and since $S$ is a martingale one
has
\begin{equation}\label{eq:5Heston} \E \left[(-Z + p_0+(\delta \cdot S)_T)^2
  \right] = \sum_{h=1}^5 \mathrm{Var} \left(-Z_i +(\delta^{(2h-1:2h)} \cdot
    S^{(2h-1:2h)})_T\right).\end{equation}
By building $\delta$ from the (discrete-time) variance optimal strategies for each of the $5$ models, one sees from \eqref{eq:5Heston} that the minimal value of \eqref{eq:objective2} over \textit{all} $\delta \in \H$ is at most $5$ times the minimal value of \eqref{eq:objective2} associated to a single Heston model. This consideration serves as a guideline for assessing the approximation quality of the neural network strategy. 

To assess the scalability of the algorithm, we now calculate the
close-to-optimal neural network hedging strategy associated to
\eqref{eq:objective2} in both instances (i.e. for $n_H = 5$ models and
for a single one, $n_H = 1$) and compare the results.  Unless
specified otherwise, the parameters are as in
Section~\ref{subsec:implementation}. Since for $n_H = 5$ we are
actually solving $5$ problems at once, we allow for a network with
more hidden nodes by taking $N_1=N_2=12 n_H$. We then train both
networks for a fixed number of time-steps (here $2\times 10^5$) and
measure the performance in terms of both training time and realized
loss (evaluated on a test set of $n_H \times 10^5$ sample paths): the
training times on a standard Lenovo X1 Carbon laptop are $5.75$ and
$2.1$ hours for $n_H = 5$ and $n_H =1$, respectively and the realized
losses are $1.13$ and $0.20$. In view of the considerations above,
this indicates that the approximation quality is roughly the same for
both instances (and close-to-optimal).

While far from a systematic study, this last example nevertheless
demonstrates the potential of the algorithm for high-dimensional
hedging problems.

\section{Disclaimer}

{\tiny
Opinions and estimates constitute our judgement as of the date of this Material, are for informational purposes only and are subject to change without notice.  This Material is not the product of J.P. Morgan’s Research Department and therefore, has not been prepared in accordance with legal requirements to promote the independence of research, including but not limited to, the prohibition on the dealing ahead of the dissemination of investment research. This Material is not intended as research, a recommendation, advice, offer or solicitation for the purchase or sale of any financial product or service, or to be used in any way for evaluating the merits of participating in any transaction. It is not a research report and is not intended as such. Past performance is not indicative of future results. Please consult your own advisors regarding legal, tax, accounting or any other aspects including suitability implications for your particular circumstances.  J.P. Morgan disclaims any responsibility or liability whatsoever for the quality, accuracy or completeness of the information herein, and for any reliance on, or use of this material in any way.  \\
\noindent
 Important disclosures at: www.jpmorgan.com/disclosures}

 \bibliographystyle{amsalpha}

\begin{thebibliography}{MHADZ93}

\bibitem[BK06]{Broadie2006}
M.~Broadie and \"O.~Kaya, \emph{Exact simulation of stochastic
  volatility and other affine jump diffusion processes}, Operations Research
  \textbf{54} (2006), no.~2, 217--231.

\bibitem[BR06]{BuRu06}
C.~Burgert and L.~R\"uschendorf, \emph{Consistent risk measures for
  portfolio vectors}, Insurance: Mathematics and Economics (2006), 289--297.

\bibitem[BTT07]{BenTal2007}
A.~Ben-Tal and M.~Teboulle, \emph{An old-new concept of convex risk
  measures: the optimized certainty equivalent}, Mathematical Finance
  \textbf{17} (2007), no.~3, 449--476.

\bibitem[Duf01]{dufresne2001}
D.~Dufresne, \emph{The integrated square-root process}, Centre for Actuarial
  Studies, University of Melbourne, 2001, Research Paper no. 90.

\bibitem[Dup94]{Dupire1994}
B.~Dupire, \emph{Pricing with a smile}, Risk \textbf{7} (1994), 18--20.

\bibitem[DZL09]{AlgoQTrading}
X.~Du, J.~Zhai, and K.~Lv, \emph{Algorithm trading using q-learning
  and recurrent reinforcement learning}, arxiv (2009),
  https://arxiv.org/pdf/1707.07338.pdf.

\bibitem[FL00]{Foellmer2000}
H.~F\"ollmer and P.~Leukert, \emph{Efficient hedging: Cost versus shortfall
  risk}, Finance and Stochastics \textbf{4} (2000), 117--146.

\bibitem[FS16]{Foellmer2016}
H.~F\"ollmer and A.~Schied, \emph{Stochastic finance: An introduction
  in discrete time}, De Gruyter, 2016.

\bibitem[Gla04]{Glasserman2004}
P.~Glasserman, \emph{Monte carlo methods in financial engineering},
  Applications of mathematics : stochastic modelling and applied probability,
  Springer, 2004.

\bibitem[GS13]{GaSch2013}
J.~Gatheral and A.~Schied, \emph{Dynamical models of market impact and
  algorithms for order execution}, Handbook on Systemic Risk (2013), 579--599.

\bibitem[Hal17]{BSQ}
I.~Halperin, \emph{Qlbs: Q-learner in the black-scholes (-merton) worlds},
  arxiv (2017), https://arxiv.org/abs/1712.04609.

\bibitem[HBP17]{Grohs2017}
G.~Kutyniok, H.~B\"olcskei, P.~Grohs and P.~Petersen,
  \emph{Optimal approximation with sparsely connected deep neural networks},
  Preprint arXiv:1705.01714 (2017).

\bibitem[HMSC95]{Soner1995}
S.~E.~Shreve, H.~M.~Soner and J.~Cvitani\'{c}, \emph{There is no nontrivial hedging
  portfolio for option pricing with transaction costs}, The Annals of Applied
  Probability \textbf{5} (1995), no.~2, 327--355.

\bibitem[HN89]{Hodges1989}
S.~Hodges and A.~Neuberger, \emph{Optimal replication of contingent
  claims under transaction costs}, The Review of Futures Markets \textbf{8}
  (1989), no.~2, 222--239.

\bibitem[Hor91]{Hornik1991}
K.~Hornik, \emph{Approximation capabilities of multilayer feedforward
  networks}, Neural Networks \textbf{4} (1991), no.~2, 251--257.

\bibitem[IAR09]{Ilhan2009}
M.~Jonsson, A.~\.{I}lhan and R.~Sircar, \emph{Optimal
  static-dynamic hedges for exotic options under convex risk measures},
  Stochastic Processes and their Applications \textbf{119} (2009), no.~10, 3608
  -- 3632.

\bibitem[IGC16]{Goodfellow2016}
Y.~Bengio, I.~Goodfellow and A.~Courville, \emph{Deep learning}, MIT
  Press, 2016, http://www.deeplearningbook.org.

\bibitem[IS15]{Ioffe2015}
S.~Ioffe and C.~Szegedy, \emph{Batch normalization: Accelerating
  deep network training by reducing internal covariate shift}, Proceedings of
  the 32nd International Conference on Machine Learning, 2015, pp.~448--456.

\bibitem[JMKS17]{MuhleKarbe2017}
M.~Reppen, J.~Muhle-Karbe and H.~M.~Soner, \emph{A primer on portfolio
  choice with small transaction costs}, Annual Review of Financial Economics
  \textbf{9} (2017), no.~1, 301--331.

\bibitem[KB15]{Kingma2015}
D.~P. Kingma and J.~Ba, \emph{Adam: a method for stochastic
  optimization}, Proceedings of the International Conference on Learning
  Representations (ICLR) (2015).

\bibitem[KMK15]{Kallsen2015}
J.~Kallsen and J.~Muhle-Karbe, \emph{Option pricing and hedging with
  small transaction costs}, Mathematical Finance \textbf{25} (2015), no.~4,
  702--723.

\bibitem[KS07]{Kloeppel2007}
S.~Kl\"oppel and M.~Schweizer, \emph{Dynamic indifference valuation
  via convex risk measures}, Mathematical Finance \textbf{17} (2007), no.~4,
  599--627.

\bibitem[LBAK10]{Andersen2010}
P.~J\"ackel, L.~B.~G.~Andersen and C.~Kahl, \emph{Simulation of
  square-root processes}, Encyclopedia of Quantitative Finance, John Wiley \&
  Sons, Ltd, 2010.

\bibitem[Lu17]{AlgoLSTM}
D.~Lu, \emph{Agent inspired trading using recurrent reinforcement learning
  and lstm neural networks}, arxiv (2017),
  https://arxiv.org/pdf/1707.07338.pdf.

\bibitem[MHADZ93]{Davis1993}
V.~G.~Panas, M.~H.~A.~Davis and T.~Zariphopoulou, \emph{European
  option pricing with transaction costs}, SIAM Journal on Control and
  Optimization \textbf{31} (1993), no.~2, 470--493.

\bibitem[MW97]{cashpf}
J.~Moody and L.~Wu, \emph{Optimization of trading systems and
  portfolios}, Proceedings of the IEEE/IAFE 1997 Computational Intelligence for
  Financial Engineering (CIFEr) (1997), 300--307.

\bibitem[PBV17]{Bank2017}
H.~M.~Soner, P.~Bank and M.~Vo{\ss}, \emph{Hedging with temporary price
  impact}, Mathematics and Financial Economics \textbf{11} (2017), no.~2,
  215--239.

\bibitem[Rog04]{Rogers2004}
L.~C.~G. Rogers, \emph{Why is the effect of proportional transaction costs
  {$O(\delta^{2/3})$}}, Mathematics of Finance (G.~Yin and Q.~Zhang, eds.),
  American Mathematical Society, Providence, RI, 2004, pp.~303--308.

\bibitem[RS10]{Rogers2010}
L.~C.~G. Rogers and S.~Singh, \emph{The cost of illiquidity and its
  effects on hedging}, Mathematical Finance \textbf{20} (2010), no.~4,
  597--615.

\bibitem[WW97]{Whalley1997}
A.~E. Whalley and P.~Wilmott, \emph{An asymptotic analysis of an optimal
  hedging model for option pricing with transaction costs}, Mathematical
  Finance \textbf{7} (1997), no.~3, 307–--324.

\bibitem[Xu06]{Xu2006}
M.~Xu, \emph{Risk measure pricing and hedging in incomplete markets},
  Annals of Finance \textbf{2} (2006), no.~1, 51--71.

\bibitem[ZJL17]{rfPf}
D.~Xu, Z.~Jiang and J.~Liang, \emph{A deep reinforcement learning
  framework for the financial portfolio management problem}, arxiv (2017),
  https://arxiv.org/abs/1706.10059.

\end{thebibliography}

\end{document}